\documentclass{article}

\usepackage[square,sort&compress,numbers]{natbib}

\usepackage[english]{babel}
\usepackage{fullpage}

\usepackage{hyperref}

\usepackage{amssymb,amsmath,amsfonts}
\usepackage{amsthm}
\usepackage[all]{xy}
\usepackage[capitalise]{cleveref}

\allowdisplaybreaks

\usepackage{xcolor}
\definecolor{MyBlue}{cmyk}{1,0.13,0,0.63}
\definecolor{MyGreen}{cmyk}{0.91,0,0.88,0.52}
\newcommand{\mylinkcolor}{MyBlue}
\newcommand{\mycitecolor}{MyGreen}
\newcommand{\myurlcolor}{webbrown}

\makeatletter
\def\@endtheorem{\endtrivlist}
\makeatother
\theoremstyle{plain}
\newtheorem{thm}{Theorem}[section]
\newtheorem{lem}[thm]{Lemma}
\newtheorem{prop}[thm]{Proposition}
\newtheorem{coro}[thm]{Corollary}

\theoremstyle{definition}
\newtheorem{defn}[thm]{Definition}
\newtheorem{remark}[thm]{Remark}

\newtheorem{example}[thm]{Example}

\newtheoremstyle{note}
{3pt}
{3pt}
{\bfseries}
{\parindent}
{\bfseries\itshape}
{:}
{.5em}
{}
\theoremstyle{note}
\newtheorem*{Question}{Question}
\newtheorem*{Problem}{Problem}
\newtheorem*{Note}{Note}

\renewcommand{\eqref}[1]{\labelcref{#1}}
\crefname{thm}{Theorem}{Theorems}
\crefname{lem}{Lemma}{Lemmas}
\crefname{prop}{Proposition}{Propositions}
\crefname{coro}{Corollary}{Corollaries}
\crefname{defn}{Definition}{Definitions}
\crefname{example}{Example}{Examples}
\crefname{remark}{Remark}{Remarks}

\hypersetup{%
  bookmarksnumbered=true,bookmarksopen=false,%
  plainpages=false,
  linktocpage=true,%
  colorlinks=true,breaklinks=true,%
  linkcolor=\mylinkcolor,citecolor=\mycitecolor,urlcolor=\myurlcolor,%
  pdfpagelayout=OneColumn,%
  pageanchor=true,%
}

\newcommand{\refcite}[2][]{\cite[#1]{#2}}
\newcommand{\refscite}[2][]{\cite[#1]{#2}}

\RequirePackage{slashed}
\newcommand{\sD}{\slashed{D}}

\RequirePackage{amssymb,amsmath}
\renewcommand{\1}{\mathbb{I}}
\newcommand{\N}{\mathbb{N}}
\newcommand{\R}{\mathbb{R}}
\newcommand{\C}{\mathbb{C}}
\newcommand{\Z}{\mathbb{Z}}
\newcommand{\A}{\mathcal{A}}
\newcommand{\mB}{\mathcal{B}}
\newcommand{\mH}{\mathcal{H}}

\newcommand{\G}{\mathcal{G}}

\newcommand{\B}{\mathcal{B}}
\newcommand{\La}{\mathcal{L}}
\newcommand{\mO}{\mathcal{O}}
\newcommand{\mU}{\mathcal{U}}

\newcommand{\E}{\mathcal{E}}
\newcommand{\mF}{\mathcal{F}}

\newcommand{\cC}{\mathcal{C}}
\newcommand{\Ob}{{\rm Ob}}
\newcommand{\Mor}{{\rm Mor}}

\newcommand{\bundlefont}[1]{{\textnormal{\texttt{#1}}}}
\newcommand{\bS}{\bundlefont{S}}
\newcommand{\bB}{\bundlefont{B}}
\newcommand{\bE}{\bundlefont{E}}
\newcommand{\bP}{\bundlefont{P}}
\newcommand{\bQ}{\bundlefont{Q}}
\newcommand{\bG}{\bundlefont{G}}
\newcommand{\bH}{\bundlefont{H}}
\newcommand{\bU}{\bundlefont{U}}
\newcommand{\bL}{\bundlefont{L}}

\newcommand{\g}{\mathfrak{g}}
\newcommand{\lu}{\mathfrak{u}}

\DeclareMathOperator{\Tr}{Tr}
\DeclareMathOperator{\tr}{tr}
\DeclareMathOperator{\Ad}{Ad}
\DeclareMathOperator{\ad}{ad}
\DeclareMathOperator{\Dom}{Dom}

\DeclareMathOperator{\Id}{Id}
\DeclareMathOperator{\id}{id}
\DeclareMathOperator{\Diff}{Diff}
\DeclareMathOperator{\End}{End}
\DeclareMathOperator{\Hom}{Hom}

\DeclareMathOperator{\cc}{c.c.}
\renewcommand{\bar}[1]{\overline{#1}}
\newcommand{\Sub}[1]{_{\scriptscriptstyle#1}}
\newcommand{\til}[1]{\widetilde{#1}}

\newcommand{\hot}{\hat\otimes}

\newcommand{\op}{{\rm op}}
\newcommand{\YM}{{\rm YM}}
\newcommand{\ED}{{\rm ED}}

\newcommand{\mattwo}[4]{
  \left(\!\!\!\begin{array}{c@{~}c}#1&#2\\#3&#4\\\end{array}\!\!\!\right)
}
\newcommand{\mat}[2]{\left(\!\!\begin{array}{#1}#2\end{array}\!\!\right)}
\newcommand{\matfour}[4]{
  \mat{c@{~}c@{~}c@{~}c}{#1\\#2\\#3\\#4\\}
}

\title{On globally non-trivial almost-commutative manifolds}
\author{Jord Boeijink$^1$\footnote{Electronic mail: \texttt{j.boeijink@math.ru.nl}}$\phantom{x}$and Koen van den Dungen$^2$\footnote{Electronic mail: \texttt{koen.vandendungen@anu.edu.au}}\\[4mm]
{\normalsize ${}^1$Institute for Mathematics, Astrophysics and Particle Physics}\\
{\normalsize Radboud University Nijmegen}\\
{\normalsize Heyendaalseweg 135, 6525 AJ Nijmegen, The Netherlands}\\[2mm]
{\normalsize ${}^2$Mathematical Sciences Institute}\\ 
{\normalsize Australian National University}\\
{\normalsize Canberra, ACT 0200, Australia}\\[2mm]
{\normalsize ${}^2$School of Mathematics and Applied Statistics}\\ 
{\normalsize University of Wollongong}\\
{\normalsize Wollongong, NSW 2522, Australia}\\
}

\begin{document}

\maketitle

\begin{abstract}
\noindent
Within the framework of Connes' noncommutative geometry, we define and study globally non-trivial (or topologically non-trivial) almost-commutative manifolds. In particular, we focus on those almost-commutative manifolds that lead to a description of a (classical) gauge theory on the underlying base manifold. Such an almost-commutative manifold is described in terms of a `principal module', which we build from a principal fibre bundle and a finite spectral triple. We also define the purely algebraic notion of `gauge modules', and show that this yields a proper subclass of the principal modules. We describe how a principal module leads to the description of a gauge theory, and we provide two basic yet illustrative examples.

\vspace{\baselineskip}\noindent

\vspace{\baselineskip}\noindent
\end{abstract}

\tableofcontents

\section{Introduction}

The framework of Connes' noncommutative geometry \cite{Connes94} provides a generalisation of ordinary Riemannian spin manifolds to noncommutative manifolds. Within this framework, the special case of a (globally trivial) almost-commutative manifold has been shown to describe a (classical) gauge theory over a Riemannian spin manifold, which ultimately led to a description of the full Standard Model of high energy physics, including the Higgs mechanism and neutrino mixing \cite{CCM07}. 

The gauge theories mentioned above are, by construction, topologically trivial (in the sense that the corresponding principal bundles are globally trivial bundles). The aim of this paper is to adapt the framework in order to allow for globally non-trivial gauge theories as well. Such a generalisation has previously been obtained for the special case of Yang-Mills theory \cite{BvS11}.

Let us briefly recall how a description of a gauge theory is obtained from an almost-commutative manifold in the globally trivial case (for a more detailed introduction we refer to e.g.\ \refcite{vdDvS12}). We start with a smooth compact $4$-dimensional Riemannian spin manifold $M$, which can be described in terms of a (real, even) spectral triple $(C^\infty(M), L^2(\bS), \sD, \gamma_5, J_M)$, where $\sD$ is the Dirac operator on the spinor bundle $\bS\to M$, $\gamma_5$ is the grading operator and $J_M$ is charge conjugation \cite{Con13}. If we take a real even finite spectral triple $(A_F, \mH_F, D_F, \gamma_F, J_F)$, one can consider the product triple
\begin{align}
\label{eq:product_triple}
M\times F := \big( C^\infty(M,A_F) , L^2(\bS)\otimes\mH_F , \sD\otimes\1 + \gamma_5\otimes D_F , \gamma_5\otimes\gamma_F , J_M\otimes J_F \big) .
\end{align}
For a real spectral triple $T = (\A, \mH, D, J)$, we define its gauge group as
\begin{align}
\label{eq:gauge_group}
\G(T) := \big\{ uJuJ^* \mid u\in\mU(\A) \big\} \simeq \mU(\A) / \mU(\A_J) ,
\end{align}
where $\A_J$ is the central subalgebra of $\A$ consisting of all elements $a\in\A$ for which $aJ=Ja^*$. 
Now suppose we have a real even finite spectral triple $F = (A_F, \mH_F, D_F, \gamma_F, J_F)$ with gauge group $G_F = \G(F)$. Then the product triple $M\times F$ defined above has gauge group $\G(M\times F) \simeq C^\infty(M,G_F)$ (at least when $M$ is simply connected), which coincides with the `classical' notion of the gauge group of the (globally trivial) principal $G_F$-bundle $\bP = M\times G_F$. (The isomorphism $\G(M\times F) \simeq C^\infty(M,G_F)$, stated in \refcite[Proposition 4.3]{BvS11} and \refcite[\S2.4.3]{vdDvS12}, is only valid under some additional conditions, and simply-connectedness of $M$ is always sufficient. We shall prove this in general for the globally non-trivial case in \cref{coro:gaugegroupsisomorphic}.)

One can show that the inner fluctuations of the operator $\sD\otimes\1 + \gamma_5\otimes D_F$ yield gauge fields (i.e.\ connection forms on the principal bundle $\bP$) as well as scalar fields (which are interpreted as Higgs fields in the noncommutative Standard Model). Finally, the spectral action principle \cite{CC97} yields a (gauge-invariant) Lagrangian from the data of the triple $M\times F$. 

This paper is organised as follows. We start in \cref{sec:prelim} by gathering some preliminary material. \cref{sec:bundles,sec:pfb,sec:conjugate} contain a brief introduction into (principal) fibre bundles and modules. In \cref{sec:covering} we describe a sufficient condition for when sections of a quotient group bundle can be lifted. Finally, we recall the basics of spectral triples and unbounded Kasparov modules in \cref{sec:spectral}. The reader who is familiar with these topics may wish to skip these preliminaries on a first reading.

In \cref{sec:acm} we describe the generalisation of the product triples $M\times F$ to (in general globally non-trivial) almost-commutative manifolds. We show that these almost-commutative manifolds are naturally given by the internal Kasparov product of an \emph{internal space} $I$ (replacing the finite spectral triple $F$) with the underlying manifold $M$. 

While every globally trivial almost-commutative manifold describes a gauge theory, this no longer holds for arbitrary globally non-trivial almost-commutative manifolds. 
In \cref{sec:gauge_modules} we therefore focus our attention on those internal spaces that will allow us to obtain a gauge theory. After briefly recalling the classification of finite spectral triples, we define the notion of a \emph{principal module}, which is an internal space built from a finite spectral triple $F$ and a principal fibre bundle $\bP$ over $M$. 
We show that the algebraic definition of the gauge group of a principal module (defined similarly to \cref{eq:gauge_group}) coincides precisely with the usual definition of the gauge group of $\bP$ (i.e.\ the vertical automorphisms of $\bP$), provided that the underlying manifold $M$ is simply connected. 

One of the main ideas in the development of noncommutative geometry has been the translation of geometric data into (operator-)algebraic data. Whereas principal modules are constructed from geometric objects (namely principal fibre bundles), we devote \cref{sct:gaugemodules} to the purely \emph{algebraic} notion of what we call a \emph{gauge module}.  We prove that these gauge modules form a proper subclass of the principal modules, which are characterised by a lift of $\bP$ to a principal $\mU(A_F)$-bundle (where $A_F$ is the algebra of the finite spectral triple $F$). 

By equipping a principal module with a connection and a `mass matrix', we construct the corresponding \emph{principal} almost-commutative manifold in \cref{sec:gauge_theory}. The remainder of this section is used to establish the main goal of this paper; namely, we describe in detail how this principal almost-commutative manifold describes a gauge theory on $M$. 

In \cref{sec:examples} we provide two basic but illustrative examples of such gauge theories, namely Yang-Mills theory and electrodynamics. The Yang-Mills example in particular shows that not every principal module is a gauge module. However, we also show that the Yang-Mills example \emph{is} a gauge module when the underlying manifold is simply-connected and $4$-dimensional. Hence on such manifolds we have no example of a principal module which is not a gauge module. 

We finish with an Outlook on possible future work. 

\subsection*{Notation}

All $C^*$-algebras and Hilbert modules will be denoted with capital letters (e.g.\ $A,B,E\ldots$), their smooth sub-algebras or pre-$C^*$-algebras (i.e.\ densely defined $*$-sub-algebras that are closed under the holomorphic functional calculus) and Hilbert pre-modules will be denoted with curly letters (e.g.\ $\A,\B,\E,\ldots$). The main exception to these conventions is the notation $\mathcal{H}$, which always denotes a complex Hilbert space. 
By $M$ we denote a smooth connected compact Riemannian spin manifold. Bundles over $M$ will be denoted with `typewriter font', where we use $\bB$ for algebra bundles, $\bE$ for vector bundles, $\bP$ for principal fibre bundles, $\bG$ for group bundles, and $\bS$ for the spinor bundle. Continuous (resp.\ smooth) sections of a bundle $\bE\to M$ will be denoted by $\Gamma(\bE)$ (resp.\ $\Gamma^\infty(\bE)$).

\section{Preliminaries}
\label{sec:prelim}

\subsection{Fibre bundles}
\label{sec:bundles}
The definitions concerning fibre bundles in this paper may differ from the definitions in some other literature, including \refcite{BvS11}, so that we find it necessary to include a list of the definitions we use. All manifolds are assumed to be \emph{smooth} and all maps between them are also assumed to be smooth. 

Let $\bE\to M$ be a smooth fibre bundle (see e.g.\ \refcite{Kobayashi-Nomizu63}). A \emph{local trivialisation} of $\bE$ is denoted by $(U,h\Sub{U})$, where $U$ is an open neighbourhood in $M$ and $h\Sub{U}\colon\pi^{-1}(U)\to U\times F$ is a diffeomorphism such that $\text{pr}_1 \circ h_U = \pi$. 
For two local trivialisations $(U,h\Sub{U})$ and $(V, h\Sub{V})$ for which $U\cap V\neq\emptyset$, we denote the corresponding \emph{transition function} by $g\Sub{VU} := h\Sub{V} \circ h\Sub{U}^{-1} \in C^\infty(U\cap V,\Diff(F))$.
\begin{defn}
\label{defn:cbundles}
Let $\cC$ be some subcategory of the category of smooth manifolds, with objects $\Ob_\cC$ and morphisms $\Mor_\cC(A,B)$ for all objects $A,B\in\Ob_\cC$. Let $M$ be a smooth manifold. A fibre bundle $\pi\colon\bE\to M$ with fibre $F$ is called a \emph{$\cC$-bundle} if $F\in\Ob_\cC$ and if on each local trivialisation $(U,h\Sub{U})$ the map $h\Sub{U}|_{\pi^{-1}(x)}\colon\pi^{-1}(x) \to x\times F$ is an isomorphism in $\Mor_\cC(\pi^{-1}(x),F)$.

Let $\pi_1\colon\bE_1\to M$ and $\pi_2\colon\bE_2\to M$ be fibre bundles. A \emph{bundle morphism} $\phi\colon\bE_1\to\bE_2$ is a smooth map such that $\pi_2\circ\phi=\pi_1$. If $\bE_1$ and $\bE_2$ are $\cC$-bundles, then $\phi$ is called a \emph{$\cC$-bundle morphism} if $\phi|_{\pi_1^{-1}(x)}\colon\pi_1^{-1}(x)\to\pi_2^{-1}(x)$ is an element of $\Mor_\cC(\pi_1^{-1}(x),\pi_2^{-1}(x))$ for each $x\in M$. 

Let $\pi\colon\bE\to M$ be a $\cC$-bundle with fibre $F$. A fibre subbundle $\pi'\colon\bE'\to M$ with fibre $F'$ is a $\cC$-subbundle if $F'\in\Ob_\cC$ and there exist local trivialisations $\{(U,h\Sub{U})\}$ for $\bE$ such that $h\Sub{U}(\bE'|\Sub{U}) \simeq U \times \iota(F')$, where $\iota$ is an injective morphism in $\Mor_\cC(F',F)$.

If $\cC$ is the category of \emph{finite-dimensional} vector spaces, \emph{finite-dimensional} ($*$-)algebras, or Lie groups, then $\cC$-bundles are referred to as vector bundles, ($*$-)algebra bundles, or group bundles (respectively). 
\end{defn}

\begin{remark}
\label{rem:algebrabundles}
 Note that according to \cref{defn:cbundles} a ($*$-)algebra bundle is always locally trivial, in contrast with the definition of ($*$-)algebra bundle in \refcite{BvS11} (where the bundle is only assumed to be locally trivial as a \emph{vector} bundle). 
 The weaker notion given in \refcite{BvS11} will here be referred to as \emph{weak ($*$-)algebra bundle}, following terminology of \refcite{Cac12}. 
\end{remark}

The space of smooth sections $\Gamma^\infty(\bE)$ of a vector bundle $\bE$ is a finitely generated projective $C^\infty(M)$-module, with pointwise addition and multiplication by $C^\infty(M)$. 
If $\phi\colon \bE_1 \rightarrow \bE_2$ is a vector bundle morphism, then
\begin{align*}
 \phi_*\colon  \Gamma^\infty(\bE_1) \rightarrow \Gamma^\infty(E_2), \quad  (\phi_*s)(x) = \phi(s(x)) 
\end{align*}
is a $C^\infty(M)$-module morphism. By the Serre-Swan theorem \cite{Swan62}, the assignment $\bE \mapsto \Gamma^\infty(\bE)$ on objects and the assigment $\phi \mapsto \phi_*$ on morphisms determines an equivalence between the category of smooth vector bundles over $M$ and the category of finitely generated projective modules over $C^\infty(M)$. 

Similarly, for a group bundle $\bG$, the sections $\Gamma^\infty(\bG)$ form a group with fibre-wise multiplication and inverse. 
\begin{example}[Unitary group bundle]
If $\bB$ is a unital $*$-algebra bundle, we define the \emph{unitary group bundle} of $\bB$ as
$$\mU(\bB) := \{b\in\bB \mid bb^*=b^*b=1\}.$$
Then $\mU(\bB)$ is a fibre subbundle of $\bB$, which forms a group bundle with group multiplication of $\mU(\bB)_x=\mU(\bB_x)$ inherited from the algebra multiplication of $\bB_x$, and group inverse given by the involution $*$. The sections of the unitary group bundle are equal to the unitary sections of the algebra bundle: $\Gamma^\infty(\mU(\bB)) = \mU(\Gamma^\infty(\bB))$.
\end{example}

\begin{example}[Endomorphism bundle]
Let $\pi_\bE\colon\bE\to M$ be a (hermitian) vector bundle with fibre $V$ and local trivialisations $(U,h\Sub{U}^\bE)$. Then the bundle of endomorphisms $\End(\bE)$ is a unital ($*$-)algebra bundle over $M$ with fibre $\End(V)$, and its local trivialisations $(U,h\Sub{U}^{\End(\bE)})$ are induced from $(U,h\Sub{U}^\bE)$. 
\end{example}

\begin{thm}[{\refcite[Theorem 3.8]{BvS11}}]
\label{thm:serreswanalgebrabundles}
Let $M$ be a compact manifold. There is an equivalence between the category of (unital) weak ($*$-)algebra bundles over $M$ and the category of (unital) (involutive) $C^\infty(M)$-module algebras that are finitely generated projective as $C^\infty(M)$-modules.
\end{thm}

We again emphasise the difference between algebra bundles and weak algebra bundles as mentioned in \cref{rem:algebrabundles}. We are grateful to Eli Hawkins who pointed out to us that a weak algebra bundle is locally trivial if and only if there exists a connection $\nabla$ satisying the Leibniz rule
\begin{align*}
\nabla(ab) = (\nabla a)b + a(\nabla b).
\end{align*}
In the continuous case, however, it remains unclear what algebraic conditions one needs to impose on a $C(M)$-module algebra $B = \Gamma(\bB)$, where $\bB$ is a (continuous) weak algebra bundle, to ensure that the weak algebra bundle $\bB$ is in fact locally trivial. 

\subsection{Principal fibre bundles and (classical) gauge theories}
\label{sec:pfb}
In this section, we briefly recall the definition of a principal fibre bundle, and some basic results. We refer to \refcite[Chapter I]{Kobayashi-Nomizu63} and \refcite{Ble81} for more details. 

\begin{defn}
\label{defn:pfb}
A \emph{principal fibre bundle} $\bP$ over $M$ with \emph{structure group} $G$ (or a \emph{principal $G$-bundle} for short) consists of a fibre bundle $\bP\xrightarrow{\pi}M$ equipped with a smooth right action of $G$ that acts freely and transitively on the fibres, such that for a local trivialisation $(U,h\Sub{U})$ of $\bP$, the map $h\Sub{U}$ intertwines the right action of $G$ on $\bP|_U$ with the natural right action of $G$ on $U\times G$. 
\end{defn}

One can construct a principal $G$-bundle $\bP$ as soon as one knows its ($G$-valued) transition functions.

\begin{thm}[Reconstruction theorem, {\refcite[Chapter I, Proposition 5.2.]{Kobayashi-Nomizu63}}]
\label{thm:reconstruction}
Let $M$ be a compact manifold, $G$ a Lie-group, and $\{U_i\}_{i\in I}$ an open covering of $M$. Suppose that for each $i,j\in I$ with $U_i \cap U_j \neq \emptyset$, there is a smooth map $g_{ij}\colon U_i \cap U_j \rightarrow G$ such that $g_{ij}(x) g_{jk}(x) g_{ki}(x) = e$ for all $x\in U_i \cap U_j \cap U_k$. Then there exists a unique principal $G$-bundle $\bP$ over $M$ with the $\{U_i\}$ as trivialising neighbourhoods and the $g_{ij}$ as transition functions.
\end{thm}

\begin{defn}
\label{defn:connection}
Let $\{(U_i,h_i)\}$ be a set of local trivialisations of $P$ such that $\cup_i U_i = M$. A connection $\omega$ on $\bP$ is a set of local $\mathfrak{g}$-valued $1$-forms $\omega_i \in \Omega^1(U_i,\g)$
such that 
\begin{align}
 \omega_j = g_{ij}^{-1} dg_{ij} + g_{ij}^{-1} \omega_i g_{ij}
\end{align}
for $i,j$ such that $U_i\cap U_j\neq\emptyset$.
\end{defn}

 \begin{defn}
Given an action $\rho$ of $G$ on a smooth manifold $F$, we define the \emph{associated bundle} $\bP\times_\rho F$ (or $\bP\times_GF$) as the quotient of the product manifold $\bP\times F$ with respect to the equivalence relation given by $(pg,f) \sim (p,\rho(g)f)$. If $F\in\Ob_\cC$ and $\rho(g)\in\Mor_\cC(F,F)$ for all $g\in G$, then $\bP\times_\rho F$ is a $\cC$-bundle. 
\end{defn}

\begin{example}
The \emph{adjoint bundle} $\Ad\bP$ is defined as the associated bundle $\bP\times_{\Ad}G$ with respect to the adjoint action $\Ad(g)h := ghg^{-1}$, $(g,h \in G)$. The adjoint bundle is a group bundle with fibres isomorphic to $G$, and its sections $\Gamma^\infty(\Ad\bP)$ then form a group with fibre-wise multiplication. 
\end{example}

\begin{defn}
\label{defn:gaugegroup}
A \emph{gauge transformation} of a principle $G$-bundle $\bP$ is a principal bundle automorphism of $\bP$ over $\text{id}:M\rightarrow M$, that is, a smooth invertible map $\phi\colon \bP\rightarrow \bP$ such that $\pi(\phi(p)) = \pi(p)$ and $\phi(pg) = \phi(p)g$ for all $p\in \bP$ and $g\in G$. The set of all such $\phi$ is called the \emph{gauge group} $\G(\bP)$ of $\bP$, where the group multiplication is given by composition. 
\end{defn}

\begin{thm}[see e.g.\ {\refcite[Ch.\ 3]{Ble81}}]
\label{thm:1}
The gauge group $\G(\bP)$ is isomorphic to the group $\Gamma^\infty(\Ad \bP)$.
\end{thm}

\begin{defn}
\label{defn:gauge_theory}
 Let $M$ be a manifold and $G$ a Lie group. A \emph{classical $G$-gauge theory} over $M$ is a principal fibre bundle $\bP$ with structure group $G$. Connections $\omega$ on $\bP$ are also called \emph{gauge potentials}.

More precisely, the bundle $\bP$ forms the \emph{setting} for a classical gauge theory. The particle fields can be described as sections of associated bundles of $\bP$. 
The description of the gauge theory is completed by specifying an \emph{action functional}, which depends on the connection and on the particle fields, and which is invariant under the action of the gauge group.
\end{defn}

\subsubsection{Structure group}
\label{sec:structure_group}

Let $\bE$ be a vector bundle with fibre $V$. 
A set of transition functions $\{(U_i,g_{ij})\}$ on $\bE$ is called a $G$-atlas if each transition function takes values in $G\subset GL(V)$. 
If $\bE$ admits a $G$-atlas, then we say that $\bE$ has \emph{structure group} $G$. Given two $G$-atlases $\{(U_i,g_{ij})\}$ and $\{(U_i,g_{ij}')\}$ (where, after taking a common refinement, we may assume without loss of generality that both atlases are given on the same open covering $\{U_i\}$), we say that they are equivalent 
if there are functions $g_i\in C^\infty(U_i,G)$ such that 
\begin{align*}
g_{ij}'(x) &= g_i(x)^{-1} g_{ij}(x) g_j(x) , & \text{for all } x\in U_i\cap U_j .
\end{align*}

Given a $G$-atlas $\{(U_i,g_{ij})\}$ on $\bE$, \cref{thm:reconstruction} constructs a unique principal $G$-bundle $\bP$, which only depends (up to isomorphism) on the equivalence class of the $G$-atlas. 
Conversely, a set of transition functions on $\bP$ uniquely determines an equivalence class of $G$-atlases on the associated bundle $\bP\times_GV$.

\begin{example}
\label{ex:unitary-princ}
Let $\bE\to M$ be a complex vector bundle with fibre $\C^N$ over a compact manifold $M$. Then all $U(N)$-atlases on $\bE$ are equivalent. Hence there is a unique (up to isomorphism) principal $U(N)$-bundle $\bP$ such that $\bE \simeq \bP\times_{U(N)}\C^N$. 
\end{example}

\begin{defn}[Lifting of structure group]
\label{defn:lift}
Let $\phi\colon H\rightarrow G$ be a surjective group homomorphism. A principal $G$-bundle $\bP\to M$ is said to lift to a principal $H$-bundle $\bQ\to M$ along $\phi$ if there is a bundle morphism $\tau\colon \bQ \rightarrow \bP$ such that $\tau(qh) = \tau(q) \phi(h)$ for all $q \in \bQ$, $h \in H$. 
Equivalently, $\bQ$ is a lift of $\bP$ if
$$
\bQ \times_\phi G \simeq \bP
$$
as principal $G$-bundles. 

If $\tau\colon \bQ \rightarrow \bP$ is such a lift and $\rho\colon G \rightarrow GL(V)$ is a finite-dimensional representation, then $\bQ \times_{\rho \circ \phi} V$ is isomorphic to $\bP \times_{\rho} V$.
We stress that a lift need not always exist, and if it exists, it need not be unique.
\end{defn}

\subsection{Conjugate modules and vector bundles}
\label{sec:conjugate}
In the construction of gauge modules in \cref{sct:gaugemodules} we will make explicit use of the notion of a conjugate module. For completeness, we recall the definition of conjugate modules and vector bundles here. Since most of the modules are endowed with a hermitian structure, we recall the definition of a hermitian module first.

\begin{defn}
 Let $\mathcal{A}$ be a $*$-algebra and let $\mathcal{E}$ be a right $\mathcal{A}$-module. A (right) \emph{hermitian structure} $(\cdot,\cdot)_\A\colon \E\times \E \rightarrow \A$ on $\E$ is a sesqui-linear map (anti-linear in the first variable) satisfying
\begin{align*}
 (e_1,e_2 a)_\A &= (e_1,e_2)_\A a; &
 (e_2,e_1)_\A &= (e_1,e_2)_\A^*; &
 (e,e)_\A &\geq 0; & 
 (e,e)_\A = 0 &\iff e =0,
\end{align*}
 for all $a \in \A$, $e_1,e_2,e \in \E$. We also write $(\cdot,\cdot)$ instead of $(\cdot,\cdot)_\A$ when no confusion can arise. A module endowed with a hermitian structure is also called a \emph{hermitian module}. A left hermitian structure ${}_\A(\cdot,\cdot)$ is defined similarly. 
 
 A hermitian structure is called \emph{non-degenerate} if the map
\begin{align*}
\E \rightarrow \E^* := \Hom_\A(\E,\A) , \qquad e_0 \mapsto ( e \mapsto  (e_0, e) )
\end{align*}
is an anti-linear isomorphism. Note that the assumption that the hermitian structure is positive-definite already implies that the map $\E\to\E^*$ is injective. Non-degeneracy therefore requires surjectivity of this map.
\end{defn}

A finitely generated projective right $\A$-module $\E$ is of the form $p\A^N$, for some $N\in\N$ and some projection $p\in M_N(\A)$. The restriction of the standard hermitian structure on $\A^N$ then gives a non-degenerate hermitian structure on $\E$. If $\A = C^\infty(M)$ (so that $\E = \Gamma^\infty(\bE)$ for some vector bundle $\bE\to M$ by the Serre-Swan theorem \cite{Swan62}), then the hermitian structure is non-degenerate if and only if it induces an inner product on each fibre of $\bE$. 

\begin{defn}
Let $\mathcal{E}$ be an $\A-\B$-bimodule with a (right) $\B$-valued hermitian structure $(\cdot,\cdot)_\B$. Its \emph{conjugate module} $\bar{\E}$ is equal to $\E$ itself as an additive group. It can naturally be endowed with a $\B-\A$-bimodule structure and a (left) $\B$-valued hermitian structure ${}_\B(\cdot,\cdot)$ by setting
\begin{align*}
 b\bar{e} &:= \bar{eb^*}, & \bar{e}a &:= \bar{a^*e}, &  {}_\B(\bar{e_1},\bar{e_2} ) &:= (e_1,e_2)_\B,
\end{align*}   
for all $a \in \A$, $b \in \B$, $e,e_1,e_2 \in \E$.
\end{defn}

If $\mathcal{E} = \Gamma^\infty(\bE)$ is the $C^\infty(M)$-module of sections of some (hermitian) vector bundle $\bE$, then the conjugate module $\bar{\mathcal{E}}$ is equal to the $C^\infty(M)$-module of sections of the conjugate vector bundle $\bar{\bE}$ which is defined as: 
\begin{defn}
 Let $\bE \rightarrow M$ be a complex vector bundle. Take $\bar{\bE}$ to be equal to $\bE$ as fibres bundles over $M$, and write $\bar{e}$ for the element in $\bar{\bE}$ that corresponds to $e \in \bE$ under this identification. The bundle $\bar{\bE}$ is turned into a vector bundle over $M$ by defining the vector space structure in $\bar{\bE}_x$ by
\begin{align*}
 (\bar{e_1}, \bar{e_2}) &\mapsto \bar{e_1 + e_2}, & \lambda \cdot \bar{e} &= \bar{\bar{\lambda} e},  
\end{align*} 
for all $\lambda \in \mathbb{C}$, $e,e_1,e_2 \in \bE_x$. The vector bundle $\bar{\bE} \rightarrow M$ is called the \emph{conjugate vector bundle} of $\bE$.
\end{defn}
The identification $\bE \ni e \mapsto \bar{e} \in \bar{\bE}$ in the above definition is an anti-linear isomorphism of vector bundles.

A local trivialisation $(U,h)$ of $\bE$ induces a local trivialisation of $\bar{\bE}$ given by the map 
\begin{align*}
\bar{h}\colon \pi^{-1}_{\bar{\bE}} (U) \ni  \bar{e} \mapsto  \bar{h e} \in U \times \bar{V}, 
\end{align*}
where $\bar{(x,v)} := (x,\bar{v}) \in U\times\bar V$. 
If $g_{ij}$ is a transition function between two local trivialisations $(U_i,h_i)$ and $(U_j,h_j)$ of $\bE$, then the transition function $\bar{g_{ij}}$ between the corresponding local trivialisations $(U_i,\bar{h_i})$ and $(U_j,\bar{h_j})$ is equal to
\begin{align}
\label{eq:transitionfunctionsconjugatevectorbundle}
 \bar{h_i}\circ \bar{h_j}^{-1}(x, \bar{v}) = \bar{h_i} \left(\bar{h^{-1}_j(x,v) }\right) = \bar{h_i h_j^{-1} (x,v)} = (x,\overline{g_{ij}(x) v}) = (x,\bar{v} \cdot g_{ij}(x)^*). 
\end{align}

From here on, we consider $\A := C^\infty(M)$. 
Suppose that $\E$ is a hermitian right $\mathcal{A}$-module with hermitian structure $(\cdot, \cdot )_\A$. 
\begin{defn}
A \emph{connection} $\nabla$ on $\E$ is a map $\nabla\colon \E \rightarrow \E \otimes_\A \Omega^1(M)$ satisfying the rule
$$
\nabla(ea) = \nabla(e) a + e\otimes da ,
$$
for all $e\in\E$ and $a\in\A$. The connection is called \emph{hermitian} if
\begin{align*}
( \nabla e_1, e_2 )_{\Omega^1(M)}  + (e_1, \nabla e_2)_{\Omega^1(M)} = d(e_1,e_2)_\A, 
\end{align*}
for all $e_1,e_2 \in \E$, where the map $( \cdot, \cdot )_{\Omega^1(M)} \colon \E \times (\E\otimes_\A\Omega^1(M)) \to \Omega^1(M)$ is defined as $(e_1,e_2\otimes\alpha)_{\Omega^1(M)} := (e_1,e_2)_\A \alpha$. 
We then define $( \cdot, \cdot )_{\Omega^1(M)} \colon (\E\otimes_\A\Omega^1(M)) \times \E \to \Omega^1(M)$ as $(e_1\otimes\alpha,e_2)_{\Omega^1(M)} := \big((e_2,e_1\otimes\alpha)_{\Omega^1(M)}\big)^*$. 
\end{defn}

The conjugate connection $\bar{\nabla}\colon\bar\E\to\Omega^1(M)\otimes_\A\bar\E$ is given by
\begin{align*}
\bar{\nabla} \bar{e} =  \bar{\nabla e}, \quad (e \in \E),
\end{align*}
where $\bar{e \otimes \omega} = \omega^* \otimes \bar{e}$ for all $e \otimes \omega \in \E \otimes_\A \Omega^1(M)$. Here $*\colon \Omega^1(M) \rightarrow \Omega^1(M)$ is defined as $(fdg)^* = f^*(dg^*)$. 
It then follows that $\bar\nabla$ is also hermitian for the map ${}_{\Omega^1(M)}( \cdot, \cdot ) \colon (\Omega^1(M) \otimes_\A\bar\E) \times \bar\E \to \Omega^1(M)$ defined as ${}_{\Omega^1(M)}(\alpha\otimes\bar{e_1},\bar{e_2}) := (e_1\otimes\alpha^*,e_2)_{\Omega^1(M)} = \alpha(e_1,e_2)_\A$.

For a commutative algebra $\mathcal{A}=C^\infty(M)$ the notion of left and right modules are equivalent. If $\mathcal{E}$ is a left $\mathcal{A}$-module with (left) $\mathcal{A}$-valued hermitian structure ${}_\mathcal{A}( \cdot, \cdot )$, then $(e_1,e_2)_\mathcal{A} := {}_\mathcal{A}(e_2,e_1)$ defines a right $\mathcal{A}$-valued hermitian structure on $\mathcal{E}$ when it is seen as a right $\mathcal{A}$-module. If $\mathcal{A}= C^\infty(M)$, we will freely use this identification.

\subsection{Covering maps}
\label{sec:covering}

We observe that, for a surjective group bundle morphism $\phi\colon\bH\to\bG$, the induced map $\phi_*\colon\Gamma^\infty(\bH)\to\Gamma^\infty(\bG)$ need \emph{not} always be surjective, as the following example shows.

\begin{example}
\label{example:notallsections}
Take $M = SO(3)$ and consider the \emph{globally trivial} group bundles $\bH = M\times U(2)$ and $\bG = M\times PSU(2)$, with the obvious group bundle morphism $\phi\colon\bH\to\bG$ given by the quotient $U(2) \to PSU(2)$. 
Since $\bH$ and $\bG$ are globally trivial, we can make the identifications $\Gamma^\infty(\bH) \simeq C^\infty(SO(3),U(2))$ and $\Gamma^\infty(\bG) \simeq C^\infty(SO(3),PSU(2))$. 
Consider the map $f\colon SO(3) \to PSU(2)$ given by the identification of $PSU(2)$ with $SO(3)$, i.e.\ $f = \text{id}$ on $SO(3)$. 
If there exists a lift $\til f\colon SO(3)\to U(2)$ such that $f = \phi\circ\til f$, then $\til f$ is nothing but a global section of the $U(1)$-principal bundle $\pi\colon U(2) \rightarrow SO(3)$. 
However, as this bundle is not globally trivial (the fundamental group of $U(2)$ is $\mathbb{Z}$, whereas the fundamental group of $SO(3) \times U(1)$ is $\mathbb{Z}_2 \times \mathbb{Z}$), such a section does not exist. Hence the map $f$, seen as a section in $\Gamma^\infty(\bG)$, is not contained in the image of $\phi_*$.
\end{example}

In this subsection we aim to find sufficient conditions for the surjectivity of $\phi_*$. In other words, we would like to have sufficient conditions to ensure that for any section $s\colon M\to \bG$ there exists a lift $\tilde{s}\colon M\to \bH$ such that $\phi_*(\tilde s) = s$. Though the existence of lifts for covering maps has been well-studied, we will typically be dealing with more general fibrations $\phi\colon \bH\to\bG$, for which the problem of existence of lifts is more complicated. We avoid this problem by reducing it to the case of covering maps, as follows.

\begin{lem}
\label{lem:subcover}
Let $p\colon E\to B$ be a fibration, and consider some map $f\colon M\to B$. Suppose there exists a submanifold $C\subset E$ such that $p|_C\colon C\to B$ is a covering space, satisfying $f_*(\pi_1(M,m)) \subset p_*(\pi_1(C,c))$, where $m\in M$ and $c\in C$ are such that $f(m) = p(c)$. Then there exists a lift $\tilde{f}\colon M\to E$ satisfying $p\circ \tilde{f} = f$ and $\tilde{f}(m) = c$. 
\end{lem}
\begin{proof}
Consider the diagram
\begin{align*}
\xymatrix{
  & C \ar@{^{(}->}[r] \ar[d]^{p|_C} & E \ar[d]^p \\
M \ar[ur]^{\tilde{f}'} \ar[r]^f & B \ar@{=}[r] & B\\
}
\end{align*}
The assumption $f_*(\pi_1(M,m)) \subset p_*(\pi_1(C,c))$ implies (see e.g.\ \refcite[Proposition 1.33]{Hatcher02}) that there exists a lift $\tilde{f}'\colon M\to C$ satisfying $\tilde{f}'(m) = c$, and then we can simply define $\tilde{f}\colon M\to E$ as the composition $M\xrightarrow{\tilde{f}'}C\hookrightarrow E$. 
\end{proof}

We now translate the above lemma into the setting of group bundles, where we will need it later.

\begin{coro}
\label{coro:surjective}
Let $M$ be a simply connected manifold, and let $\bG,\bH$ be group bundles over $M$. If $\bG$ is covered by a subbundle $\bU$ of $\bH$ via a group bundle morphism $\phi\colon \bH \rightarrow \bG$, then the map $\phi_*\colon \Gamma^\infty(\bH) \rightarrow \Gamma^\infty(\bG)$, given by $s \mapsto \phi \circ s$, is surjective.  
\end{coro}
\begin{proof}
By assumption, $\phi|_\bU\colon\bU\to\bG$ is a covering space. 
Since $\pi_1(M,m)$ is trivial (by definition of simply-connectedness) it follows from \cref{lem:subcover} that each section $s\colon M\to \bG$ can be lifted to a section $\widetilde{s}\colon M\to \bU\subset\bH$ such that $\phi_*(\widetilde{s}) = s$. 
\end{proof}

\subsection{Spectral triples and Kasparov modules}
\label{sec:spectral}

Spectral triples were introduced in \refcite{Connes94} as a noncommutative analogue of a spin manifold. 

\begin{defn}
\label{defn:spectral_triple}
A \emph{spectral triple} $(\A, \mH, D)$ is given by an involutive unital algebra $\A$ represented (faithfully) as bounded operators on a Hilbert space $\mH$ and a self-adjoint (generally unbounded) operator $D$ with compact resolvent (or equivalently,\ $(1+D^2)^{-1/2}$ is a compact operator) such that $a\cdot\Dom D\subset\Dom D$ and the commutator $[D,a]$ is bounded for each $a\in\A$. 

A spectral triple is called \emph{even} if there exists a $\Z_2$-grading $\gamma$ on $\mH$ that commutes with any $a\in\A$ and anti-commutes with $D$.

A spectral triple is called \emph{real} if there exists an anti-unitary isomorphism $J\colon \mH\rightarrow\mH$ satisfying
\begin{align*}
J^2 &= \varepsilon , & JD &= \varepsilon' DJ , & J\gamma &= \varepsilon'' \gamma J \text{ (if $\gamma$ exists)} ,\\
[a,JbJ^*] &= 0 , & [[D,a],JbJ^*] &= 0 , & &\forall a,b\in\A . 
\end{align*}
The signs $\varepsilon$, $\varepsilon'$ and $\varepsilon''$ determine the \emph{KO-dimension} $n$ modulo $8$ of the real spectral triple, according to the following table:
\begin{align*}
\begin{array}{c|cccccccc}
n & 0 & 1 & 2 & 3 & 4 & 5 & 6 & 7 \\
\hline
\varepsilon & 1 & 1 & -1 & -1 & -1 & -1 & 1 & 1 \\
\varepsilon' & 1 & -1 & 1 & 1 & 1 & -1 & 1 & 1 \\
\varepsilon'' & 1 & & -1 &  & 1 & & -1 & \\
\end{array}
\end{align*}
\end{defn}
We will refer to the conditions $[a,JbJ^*] = 0$ and $[[D,a],JbJ^*] = 0$ as the zeroth- and first-order condition, respectively.

Given an algebra $\A$, we define the \emph{opposite algebra} as the vector space $\A^\op := \{a^\op \mid a\in\A\}$ with the \emph{opposite product} $a^\op b^\op = (ba)^\op$. For a real spectral triple, we therefore have a linear representation of $\A^\op$ on $\mH$ given by $a^\op \mapsto Ja^*J^*$. 

The notion of spectral triple can be seen as an unbounded version of a Fredholm module. The generalisation of Fredholm modules from Hilbert spaces to Hilbert modules was performed by Kasparov \cite{Kasparov80}, where for any two graded $C^*$-algebras $A$ and $B$ the set $KK(A,B)$ was defined as the set of equivalence classes of certain Kasparov $A-B$-modules. In addition, there exists a Kasparov product $KK(A,B) \times KK(B,C) \to KK(A,C)$. More details can be found in e.g.\ \refcite{Blackadar98}. 
Kasparov modules were subsequently generalised to the unbounded picture by Baaj and Julg \cite{BJ83}. In this paper we will only focus on the unbounded picture, which we briefly recall below. 

\begin{defn}[\refcite{BJ83}]
\label{defn:Kasp_mod}
Given $\Z_2$-graded $C^*$-algebras $A$ and $B$, an \emph{unbounded Kasparov $A-B$-module} $({}_{\phi(A)}E_B,D)$ is given by
\begin{itemize}
\item a $\Z_2$-graded, countably generated, right Hilbert $B$-module $E_B$;
\item a $\Z_2$-graded $*$-homomorphism $\phi\colon A\to\End_B(E)$;
\item a self-adjoint, regular, odd operator $D\colon\Dom D\subset E\to E$ such that, for all $a$ in a dense sub-algebra $\A$ of $A$, $\phi(a)\cdot\Dom D\subset \Dom D$ and $[D,\phi(a)]_\pm$ is (or extends to) a bounded endomorphism, and $\phi(a)(1+D^2)^{-\frac{1}{2}}$ is a compact endomorphism (i.e.\ it lies in $\End^0_B(E)$). 
\end{itemize}
The set of all unbounded Kasparov $A-B$-modules is denoted by $\Psi(A,B)$. We will often simply write ${}_AE_B$ instead of ${}_{\phi(A)}E_B$.
\end{defn}

A right Hilbert $\C$-module is just a Hilbert space. A spectral triple $(\A,\mH,D)$ may then be seen as an unbounded Kasparov $A-\C$-module $({}_A\mH_\C,D)$, where the $C^*$-closure $A$ of $\A$ is trivially graded.

There is a natural map from the unbounded picture to the bounded one. This map is defined by replacing the operator $D$ in $({}_{\phi(A)} E_B, D)$ by $b(D) = D(1+D^2)^{-\frac{1}{2}}$, where $b\colon \mathbb{R} \rightarrow \mathbb{R}$ denotes the function $b(x) = x(1+x^2)^{-\frac12}$.

\begin{thm}[{\refcite[Theorems 17.10.7 and 17.11.4]{Blackadar98}}]
 If $({}_{\phi(A)} E_B, D) \in \Psi(A,B)$, then $({}_{\phi(A)} E_B, b(D)) \in KK(A,B)$. Moreover, if $A$ is separable and $B$ is $\sigma$-unital, then this map $\Psi(A,B) \rightarrow KK(A,B)$ is surjective.
\end{thm}

The Kasparov product has an unbounded analogue. To be precise, we say that an unbounded Kasparov $A-C$-module $({}_{\phi(A)}E_C,D)$ represents the Kasparov product of two unbounded Kasparov modules $({}_{\phi_1(A)}{E_1}_B,D_1)$ and $({}_{\phi_2(B)}{E_2}_C,D_2)$ if $[(E,b(D))]\in KK(A,C)$ is the Kasparov product of $[(E_1,b(D_1))]\in KK(A,B)$ and $[(E_2,b(D_2))]\in KK(B,C)$, where the square brackets indicate that we take the equivalence class of the Kasparov-module.

We will show in \cref{sec:acm} that the construction of an almost-commutative manifold as the product of an internal space $I$ with the underlying manifold $M$ corresponds to an unbounded Kasparov product on the level of $KK$-classes. Although this follows from the (more general) framework of Mesland \cite{Mes14}, we will prove it directly using the following result. 

\begin{thm}[Kucerovsky, \refcite{Kuc97}]
\label{thm:Kucerovsky}
Let $({}_{\phi_1(A)} E^1_B,D_1)$ and $({}_{\phi_2(B)} E^2_C,D_2)$ be unbounded Kasparov modules. Write $E := E^1\hot_B E^2$, where $\hat{\otimes}$ denotes the \emph{graded} tensor product. Suppose that $({}_{\phi_1(A) \otimes \text{id}} E_C,D)$ is an unbounded Kasparov module such that:
\begin{itemize}
\item[i)] for all $e_1$ in a dense subspace of $\phi_1(A) E^1$, the commutators 
$$
\left[\mat{cc}{D&0\\0&D_2\\},\mat{cc}{0&T_{e_1}\\T_{e_1}^*&0\\}\right]
$$
are bounded on $\Dom(D\oplus D_2)\subset E\oplus E^2$, where $T_{e_1}\colon E^2 \rightarrow E$ is given by $T_{e_1}(e_2) = e_1 \otimes e_2$;
\item[ii)] $\Dom(D)\subset\Dom(D_1\hot1)$;
\item[iii)] $((D_1\hot1)e|De) + (De|(D_1\hot1)e) \geq K (e|e)$ for some $K\in\R$, for all $e\in\Dom(D)$. 
\end{itemize}
Then $({}_{\phi_1(A) \otimes \text{id}} E_C,D)$ represents the Kasparov product of $({}_{\phi_1(A)} E^1_B,D_1)$ and $({}_{\phi_2(B)} E^2_C,D_2)$. 
\end{thm}

\section{Almost-commutative manifolds}
\label{sec:acm}

Almost-commutative manifolds $M\times F$ of the form \cref{eq:product_triple} were first studied in \refcite{CL91} and \refscite{D-VKM1,D-VKM2,D-VKM3,D-VKM4}. They were later used in \refscite{Connes96,CCM07} to geometrically describe Yang-Mills theories and the Standard Model of elementary particles. The name almost-commutative manifolds was coined in \refcite{Class_IrrACG_I}, their classification started in \refscite{Kra98,PS98}.

Let $M$ be a smooth compact even-dimensional Riemannian spin manifold. We assume (throughout this section) that $M$ has dimension $4$. The manifold $M$ can be completely characterised \cite{Con13} by the real even spectral triple 
$$
(C^\infty(M), L^2(\bS), \sD, \gamma_5, J_M) ,
$$
which is often referred to as the \emph{canonical} spectral triple for $M$. 
Here $\bS$ is a spinor bundle over $M$, $\sD = -i c \circ \nabla^\bS$ is the corresponding Dirac operator (where $\nabla^\bS$ is the lift of the Levi-Civita connection on $M$, and $c$ denotes Clifford multiplication with the conventions $c(v)c(w)+c(w)c(v) = 2g(v,w)$ and $c(v)^* = c(v)$ for any $v,w\in\Gamma^\infty(T^*M)$), $\gamma_5$ is the grading of the spinor bundle, and $J_M$ is the charge conjugation operator. Given a real even \emph{finite} spectral triple $(A_F, \mH_F, D_F, \gamma_F, J_F)$ (for which $\dim\mH_F<\infty$), we can construct the product triple
$$
M\times F := \left( C^\infty(M,A_F) , L^2(\bS)\otimes\mH_F , \sD\otimes\1 + \gamma_5\otimes D_F , \gamma_5\otimes\gamma_F , J_M\otimes J_F \right) .
$$
Defining the (globally trivial) algebra bundle $\bB = M\times A_F$ and the (globally trivial) vector bundle $\bE = M\times\mH_F$, we can rewrite $C^\infty(M,A_F) \simeq \Gamma^\infty(\bB)$ and $L^2(\bS)\otimes\mH_F \simeq L^2(\bS\otimes\bE)$. The purpose of this section is to generalise the construction of $M\times F$ to \emph{globally non-trivial} bundles over $M$. 
At the same time, we will put this generalised construction in the context of the Kasparov product between unbounded Kasparov modules. The globally non-trivial case was first considered in \refcite{BvS11} for the case of algebra bundles with fibre $M_N(\C)$, and has also been studied more generally in \refcite{Cac12}.

\subsection{The internal space}

\begin{defn}
A (smooth) \emph{internal space} $I^\infty$ over a compact manifold $M$ is given by the data
\begin{align*}
I^\infty := \left( \Gamma^\infty(\bB),\; \Gamma^\infty(\bE),\; D_I\right) ,
\end{align*}
where $\bE$ is a hermitian vector bundle over $M$, $\bB$ is a unital $*$-algebra subbundle of $\End(\bE)$, and $D_I$ is a hermitian element of $\Gamma^\infty(\End(\bE)) \simeq \End_{C^\infty(M)}(\Gamma^\infty(\bE))$. 

An internal space is called \emph{even} if there is a grading $\gamma_I$, i.e.\ an endomorphism $\gamma_I\in\Gamma^\infty(\End(\bE))$ such that 
\begin{align*}
\gamma_I^* &= \gamma_I , & \gamma_I^2 &= 1 , & \gamma_I D_I &= -D_I \gamma_I , & \gamma_I a &= a \gamma_I \quad \forall a\in\Gamma^\infty(\bB) .
\end{align*}
An even internal space is called \emph{real} if there is a real structure $J_I$, i.e.\ an anti-unitary endomorphism $J_I$ on $\bE$ such that 
\begin{align*}
J_I^2 &= \varepsilon , & J_I D_I &= \varepsilon' D_I J_I , & J_I \gamma_I &= \varepsilon'' \gamma_I J_I ,\\
[a,Jb^*J^*] &= 0 , & \big[[D_I,a],Jb^*J^*\big] &= 0 , & &\forall a,b\in\Gamma^\infty(\bB),
\end{align*}
where the signs determine the KO-dimension of the internal space according to the same table as in \cref{defn:spectral_triple}. 
\end{defn}

\begin{remark}
\label{remark:varying_mass}
The endomorphism $D_I$ will be interpreted as a mass matrix describing the masses of the elementary particles. 
We would like to point out a few things about this mass matrix:
\begin{enumerate}
\item On a local trivialisation (say, around a point $x\in M$) we can view the endomorphism $D_I$ as a matrix-valued function $D_I(x)$, but the precise form of this matrix $D_I(x)$ depends on the choice of local trivialisation. 
However, since the transition functions are unitary, two different choices of local trivialisations yield two unitarily equivalent mass matrices, and hence the eigenvalues of the matrix $D_I(x)$ (i.e.\ the masses of the particles) are independent of the choice of local trivialisation. 
\item These eigenvalues of $D_I(x)$ are (by default) allowed to vary as a function of $x\in M$. In the standard (globally trivial) approach one can also make the (ad hoc) decision to promote the mass parameters to functions (although this is usually not done). However, this would be unnatural from the perspective that a (globally trivial) almost-commutative manifold is the (external) Kasparov product of a Riemannian spin manifold with a finite spectral triple. 
Instead, varying mass parameters are more naturally described by replacing the finite spectral triple by an internal space (which works equally well in the globally trivial case) and replacing the external by the internal Kasparov product. 
As such, the promotion of the mass parameters to functions becomes a natural attribute of our framework. 
\item One could ask whether it is always possible to choose these mass parameters to be globally constant (as in the usual approach). We expect that this might not always be possible in the general globally non-trivial case, but it is unclear what the precise topological obstructions would be.
\end{enumerate}
\end{remark}

We shall write $\A = C^\infty(M)$, $\B = \Gamma^\infty(\bB)$, and $\E = \Gamma^\infty(\bE)$. Their respective $C^*$-closures are denoted by $A = C(M)$, $B = \Gamma(\bB)$, and $E = \Gamma(\bE)$. 

\begin{prop}
An even internal space $I^\infty = (\Gamma^\infty(\bB), \Gamma^\infty(\bE), D_I)$ yields an unbounded Kasparov $B-A$-module $I = ({}_B \Gamma(\bE)_{A}, D_I)$. 
\end{prop}
\begin{proof}
The algebras $A$ and $B$ are trivially graded $C^*$-algebras, and $E=\Gamma(\bE)$ is a $\Z_2$-graded, finitely generated, right Hilbert $A$-module, with a left action of $B$ that commutes with the (right) action of $A$. The properties of $\gamma_I$ guarantee that all conditions with respect to the grading are satisfied. 
For instance, the condition $(E^{(m)}, E^{(n)} ) \subset A^{(m+n)}$, where $m,n \in \mathbb{Z}_2$, is satisfied, since the condition $\gamma_I^* = \gamma_I$ implies that $\langle s,t \rangle = 0$ as soon as one of the arguments is odd and the other is even. The operator $D_I$ is a bounded, self-adjoint, odd operator by definition (and hence it is automatically regular). The boundedness of $D_I$ implies that $[D_I,b]$ is also bounded for all $b\in \B$. 

For a compact manifold $M$ the compact endomorphisms of the $C(M)$-module $\Gamma(\bE)$ are exactly the sections of the endomorphism bundle $\text{End }\bE$, i.e. $\text{End}^0_{C(M)} (\Gamma(\bE)) = \Gamma(\End(\bE))$ (since $\Gamma(\End(\bE))$ is already unital, the compact endomorphisms of $\Gamma(\bE)$ are actually all the bounded endomorphisms, see e.g.\ \refcite[Proposition 3.9]{GVF01}).
Thus, $b(1 + D_I^2)^{-\frac{1}{2}}$ is compact for all $b \in \B$, because both $(1 + D_I^2)^{-\frac{1}{2}}$ and $b$ are compact.
Hence $({}_B\Gamma(\bE)_{A}, D_I)$ has all the properties mentioned in \cref{defn:Kasp_mod}. 
\end{proof}

\subsection{The product space}

\begin{defn}
\label{defn:ACM}
Let $I^\infty := (\Gamma^\infty(\bB), \Gamma^\infty(\bE), D_I, \gamma_I, J_I)$ be a real even internal space over $M$, with $M$ a compact $4$-dimensional Riemannian spin manifold. Let $\nabla^I$ be a hermitian connection on $\bE$. We define a real even \emph{almost-commutative manifold} to be
\begin{align*}
I^\infty\times_\nabla M := \left( \Gamma^\infty(\bB),\; L^2(\bE\otimes \bS),\; \sD_\bE + D_I\otimes\gamma_5, \gamma_I\otimes\gamma_5, J_I\otimes J_M \right) ,
\end{align*}
where $L^2(\bE\otimes \bS) \simeq \Gamma(\bE)\otimes_{C(M)} L^2(\bS)$ are the $L^2$-sections of the twisted spinor bundle $\bE\otimes \bS$, and $\sD_\bE$ is the twisted Dirac operator
$$
\sD_\bE := \1\otimes_\nabla\sD := \1\otimes\sD -i (\1\otimes c)\circ(\nabla^I\otimes\1)  .
$$
Note that by definition the underlying manifold of an almost-commutative manifold is always assumed to be of dimension $4$.
\end{defn}
We note that our definition of almost-commutative manifolds fits within the slightly more general definition of almost-commutative spectral triples given in \refcite[Definition 2.3]{Cac12}. 

The order of $I^\infty$ and $M$ in the notation $I^\infty \times_\nabla M$ is reversed in comparison with the order of $F$ and $M$ in $M \times F$. The reason is that the order $I^\infty \times_\nabla M$ is more natural from a $KK$-theoretical viewpoint, whereas the notation $M \times F$ for the globally trivial case is quite standard in the literature. In the remainder of this section, we show in detail that an almost-commutative manifold $I^\infty\times_\nabla M$ determines an unbounded Kasparov $B-\C$-module (i.e.\ a spectral triple over $\B$) whose $KK$-class represents the Kasparov product  between the $KK$-classes of the internal space $I^\infty$ and the canonical spectral triple for $M$.

\begin{prop}
\label{prop:spec_trip}
Let $I^\infty = (\Gamma^\infty(\bB), \Gamma^\infty(\bE), D_I, \gamma_I, J_I)$ be a real even internal space over a compact Riemannian spin manifold $M$ of even $KO$-dimension $k$. Let $\nabla^I$ be a hermitian connection on $\bE$ that commutes with the grading $\gamma_I$, satisfies $\nabla^I_\mu J_I = J_I \nabla^I_\mu$, and is such that the induced connection $[\nabla^I,\cdot]$ on $\End \bE$ restricts to a connection on $\bB$.
Then the real even almost-commutative manifold $I^\infty\times_\nabla M$ is a real even spectral triple of $KO$-dimension $4 + k$ (mod $8$).
\end{prop}
\begin{proof}
Let us write $D := \sD_\bE + D_I\otimes\gamma_5$. We need to show that $[D,a]$ is bounded for all $a\in\Gamma^\infty(\bB)$. Since $D_I$ is bounded itself, we need only check this for the twisted Dirac operator $\sD_\bE$, and we find
$$
[\sD_\bE,a] = -i c([\nabla^I,a]) ,
$$
where, with some abuse of notation, we write $c(T\otimes\alpha) = T\otimes c(\alpha)$ for $T\in\Gamma^\infty(\End\bE)$ and $\alpha\in\Omega^1(M)$. Hence for smooth $a$ the commutator $[\sD_\bE,a]$ indeed acts as a bounded operator on $L^2(\bE\otimes\bS)$. Furthermore we need to show that $D$ has compact resolvent, and (as $M$ is compact) for this it is sufficient to show that $D^2$ (and hence $D$) is elliptic. The Lichnerowicz-Weitzenb\"ock formula shows that the square of the twisted Dirac operator $\sD_\bE$ is a generalised Laplacian, and hence is elliptic. The bounded (zeroth-order) perturbation $\sD_\bE\to\sD_\bE+D_I\otimes\gamma_5$ does not affect this ellipticity. Hence $I^\infty\times_\nabla M$ is indeed a spectral triple. 

Given the grading operators $\gamma_I$ and $\gamma_5$, it is straightforward to check that $D(\gamma_I\otimes\gamma_5)=-(\gamma_I\otimes\gamma_5)D$, provided that $[\nabla^I,\gamma_I]=0$. 

Given the real structures $J_I$ and $J_M$, the operator $J_I\otimes J_M$ is anti-unitary and satisfies
\begin{gather}
\label{eq:KO-twist}
(J_I\otimes J_M)^2 = -\varepsilon , \qquad\qquad D(J_I\otimes J_M) = (J_I\otimes J_M)D , \nonumber\\ 
(J_I\otimes J_M)(\gamma_I\otimes\gamma_5) = \varepsilon''(\gamma_I\otimes\gamma_5)(J_I\otimes J_M) ,
\end{gather}
where the signs $\varepsilon,\varepsilon''$ are determined by the KO-dimension $k$ of $J_I$. The first equality in \cref{eq:KO-twist} is immediate from $J_M^2=-1$ and $J_I^2=\varepsilon$. Using the relations
\begin{align*}
J_M\sD&=\sD J_M , & \gamma^\mu J_M &= -J_M \gamma^\mu , & \gamma_5 J_M &= J_M\gamma_5 , & J_ID_I&=D_IJ_I , & \nabla^I_\mu J_I &= J_I \nabla^I_\mu ,
\end{align*}
the second equality in \cref{eq:KO-twist} is checked by a local calculation (writing $(\1\otimes c)\circ(\nabla^I\otimes\1) = \nabla^I_\mu\otimes\gamma^\mu$): 
\begin{align*}
D(J_I\otimes J_M)(s\otimes\psi) &= (J_Is)\otimes(\sD J_M\psi) -i (\nabla^I_\mu J_Is)\otimes(\gamma^\mu J_M\psi) + (D_IJ_Is)\otimes(\gamma_5 J_M\psi) \\
&= (J_Is)\otimes(J_M\sD\psi) + i (J_I\nabla^I_\mu s)\otimes(J_M\gamma^\mu\psi) + (J_ID_Is)\otimes(J_M\gamma_5\psi) \\
&= (J_Is)\otimes(J_M\sD\psi) - (J_I \nabla^I_\mu s)\otimes(J_M i\gamma^\mu\psi) + (J_ID_Is)\otimes(J_M\gamma_5\psi) \\
&=(J_I\otimes J_M)D(s\otimes\psi) .
\end{align*}
The third equality in \cref{eq:KO-twist} immediately follows from $[J_M,\gamma_5]=0$ and $J_I\gamma_I=\varepsilon''\gamma_IJ_I$. From the values of $-\varepsilon$ and $\varepsilon''$ it is immediate that the $KO$-dimension of $I^\infty \times M$ should be $4 +k$ (mod $8$) (see the table in \cref{defn:spectral_triple}).

The zeroth-order condition on $I^\infty \times_\nabla M$ is immediate from the zeroth-order condition on $I^\infty$. Moreover, 
\begin{align*}
[[\sD_\bE,a], JbJ^* ] = -i [c( [\nabla^I,a]), JbJ^*] = -i c([[\nabla^I,a],JbJ^*]) = 0,
\end{align*}
because, by assumption, $[\nabla^I,a] \in \Gamma^\infty(\bB) \otimes_{C^\infty(M)} \Omega^1(M)$, which commutes with $JbJ^*$. Together with the first-order condition on $D_I$, this implies that $D$ satisfies the first-order condition.
\end{proof}

For a real spectral triple $T = (\A,\mH,D,J)$, the gauge group is defined in \refcite[Definition 2.5]{vdDvS12} as
\begin{align}
\label{eq:gaugegroup_ST}
\G(T) := \big\{ uJuJ^* \mid u\in\mU(\A) \big\} \simeq \mU(\A) / \mU(\A_J) ,
\end{align}
where the central subalgebra $\A_J$ is defined as $\A_J := \{a\in\A \mid aJ = Ja^*\}$. For the above almost-commutative manifold, we therefore obtain the gauge group
$$
\G(I^\infty\times_\nabla M) = \mU(\B) / \mU(\B_J) ,
$$
for the real structure $J = J_I\otimes J_M$. However, since $\B_J \simeq \B_{J_I}$, we find that the gauge group of the almost-commutative manifold is completely determined by the internal space, and we write
\begin{align}
\label{eq:gaugegroup_IS}
\G(I^\infty\times_\nabla M) \simeq \G(I^\infty) := \{ uJ_IuJ_I^* \mid u\in\mU(\B) \} .
\end{align}

\subsection{The Kasparov product}
We now show that the product $I^\infty \times_\nabla M$ is an unbounded representative for the Kasparov product of the $KK$-classes of $I^\infty$ and the canonical spectral triple for $M$. We first prove this for the cases where $D_I = 0$, and then show that the presence of $D_I$ is irrelevant at the level of $KK$-classes. 

Let $I^\infty$ be an internal space over $M$, where $D_I=0$, and consider the unbounded Kasparov module $I := ({}_BE_A,0)$, where $E = \Gamma(\bE)$.
We know from \cref{prop:spec_trip} that $I^\infty\times_\nabla M = (\B, L^2(\bE\otimes\bS), D)$ is a spectral triple, which thus yields an unbounded Kasparov module $I\times_\nabla M = ({}_B L^2(\bE\otimes\bS)_\C, D) \in \Psi(B,\mathbb{C})$ (\cref{defn:Kasp_mod}).

\begin{prop}
\label{prop:Kasp_prod_0}
 The unbounded Kasparov module $I\times_\nabla M$ represents the Kasparov product of (the classes of) $I\in\Psi(B,A)$ and $({}_AL^2(\bS)_\C,\sD)\in\Psi(A,\C)$. 
\end{prop}
\begin{proof}
It suffices to check the conditions of \cref{thm:Kucerovsky}. Since $D_I =0 $, conditions ii) and iii) are trivial, and we only need to check Condition i). For all $e$ in a dense subspace of $BE =E$, we need to check boundedness of
\begin{align*}
&D T_e - T_e \sD && \text{on } \Dom(\sD) \subset L^2(\bS) , \\
&\sD T_e^* - T_e^* D && \text{on } \Dom(D) \subset E\otimes_A L^2(\bS) \simeq L^2(\bE\otimes\bS) ,
\end{align*}
where $D = \sD_\bE = -i (\1\otimes c)\circ(\1\otimes\nabla^\bS + \nabla^I\otimes\1)$. For $\psi\in\Dom(\sD)$ we obtain
\begin{align*}
(DT_e-T_e\sD)\psi &= -i (\1\otimes c)\circ(\1\otimes\nabla^\bS + \nabla^I\otimes\1) e\otimes\psi - e\otimes \sD\psi = -i c(\nabla^Ie)\otimes\psi ,
\end{align*}
which is indeed bounded for all $e$ in the dense subspace $\E$. Next, for $f\otimes\psi\in\Gamma^\infty(\bB \otimes \bS)\subset\Dom(D)$ we obtain
\begin{align*}
(\sD T_e^* - T_e^* D) (f\otimes\psi) &= \sD (e|f) \psi - (e|f) \sD\psi +i \big(e|c(\nabla^If)\big) \psi = -i c\big( \nabla^Ie|f\big) \psi ,
\end{align*}
where we have used the compatibility of the connection $\nabla^I$ with the hermitian form $(\cdot|\cdot)_\A$, and so $\sD T_e^* - T_e^* D$ is a zeroth-order differential operator for smooth $e$. 
\end{proof}

To prove a similar result for the case where $D_I \neq 0$, we use the following two lemmas.

\begin{lem}
\label{lem:bdd_fgp}
If ${}_{\phi(B)} E_A$ is finitely generated projective as a right $A$-module, then for any self-adjoint, odd endomorphism $F\in\End_A(E)$, the unbounded Kasparov $B-A$-modules $({}_{\phi(B)} E_A, F)$ and $({}_{\phi(B)} E_A, 0)$ represent the same class in $KK(B,A)$.  
\end{lem}
\begin{proof}
Since $E$ is a finitely generated projective $A$-module, all bounded endomorphisms are in fact compact, i.e.\ $\End_A(E) = \End_A^0(E)$. The equivalence of the compact operators $0$ and $b(F) = F(1+F^2)^{-\frac12}$ is then simply obtained via the operator homotopy $t \mapsto tb(F)$, for $t\in[0,1]$. Hence the modules $({}_{\phi(B)} E_A, b(F))$ and $({}_{\phi(B)} E_A, 0)$ are equivalent bounded Kasparov $B-A$-modules. 
\end{proof}

\begin{lem}[{see also \refcite[Corollary 17]{Kuc97}}]
\label{lem:boundedperturbation}
Let $({}_{\phi(B)}E_A, D) \in \Psi(B,A)$ and let $T \in \End_A(E)$ be self-adjoint and odd. Then
\begin{enumerate}
\item $({}_{\phi(B)} E_A, D + T)$ is also an unbounded Kasparov module in $\Psi(B,A)$, and;
\item $({}_{\phi(B)} E_A, D + T)$ and $({}_{\phi(B)} E_A, D)$ represent the same class in $KK(B,A)$.
\end{enumerate}
\end{lem}
\begin{proof}
\begin{enumerate}
\item 
Since $T$ is bounded and self-adjoint, it follows from the Kato-Rellich theorem for Hilbert modules (see \refcite[Theorem 4.5]{KL12}) that the sum $D + T$ remains self-adjoint \emph{and} regular.
The only non-trivial thing to prove is that $D+T$ has compact resolvent, i.e.\ $\phi(b)(1+(D+T)^2)^{-1/2} \in\End_A^0(E)$ for all $b\in\B\subset B$. This is equivalent to showing that $\phi(b)(\pm i+D+T)^{-1}$ is compact. The operator $(\pm i+D+T)^{-1}$ maps $E$ into $\Dom (D+T) = \Dom D$, so that $(\pm i+D)(\pm i+D+T)^{-1}$ is a well-defined bounded operator on $E$. From
$$
\phi(b)(\pm i+D+T)^{-1} = \phi(b)(\pm i+D)^{-1}(\pm i+D)(\pm i+D+T)^{-1}
$$
we then see that $\phi(b)(\pm i+D+T)^{-1}$ is compact. 
\item The idea is to prove that $({}_{\phi(B)} E_A, D + T) \in \Psi(B,A)$ represents the Kasparov product $[({}_{\phi(B)} E_A, D)] \otimes_A [({}_A A_A, 0)]$. It is enough to show that all the conditions in \cref{thm:Kucerovsky} are satisfied. First of all,
\begin{align*}
 A\ni a \mapsto (D+T)T_e(a) = ((D + T)e)a, \quad (f \otimes a) \mapsto T_e^* (D + T)(f \otimes a) = ((D+T)e,f)_A a,
\end{align*}
are both clearly bounded on $A$ and $\Dom (D + T)= \Dom D$, respectively, for all $e\in \Dom D$. 
In particular, this holds for all $e\in\phi(\B)\Dom D$, which is a dense subset of $\phi(B)E$. 
This proves that Condition (i) in \cref{thm:Kucerovsky} is satisfied.

Since $\text{Dom }D = \text{Dom}(D + T)$, Condition (ii) is also satisfied. For the final condition, a small calculation shows that
\begin{align*}
&( (D+T)e, De ) + ( De, (D+T)e  ) \\
&\qquad= ( (D+T)e, (D+T)e ) - ((D+T)e, Te) + (De, (D+T)e) \\ 
&\qquad= ( (D+T)e, (D+T)e )  - (Te,Te) + (De,De) \geq - \|T\|^2 (e,e),
\end{align*}
for all $e \in \text{Dom }D$, since $( (D+T)e, (D+T)e )$ and $(De,De)$ are positive. \qedhere
\end{enumerate}
\end{proof}

\begin{coro}
The unbounded Kasparov module $I\times_\nabla M = ({}_BE\otimes_AL^2(\bS)_\C, \1\otimes_\nabla\sD + D_I \otimes\gamma_5)$ represents the Kasparov product of $I = ({}_BE_A, D_I)$ with $({}_AL^2(\bS)_\C, \sD)$.
\end{coro}
\begin{proof}
By \cref{lem:bdd_fgp} we know that $({}_BE_A, D_I)$ and $({}_BE_A, 0)$ represent the same Kasparov class. From \cref{prop:Kasp_prod_0} it then follows that the cycle $({}_BE\otimes_A L^2(\bS)_\C, \sD_\bE)$ also represents the Kasparov product of $({}_BE_A, D_I)$ with $({}_AL^2(\bS)_\C, \sD)$. According to \cref{lem:boundedperturbation} the cycle $({}_BE\otimes_A L^2(\bS)_\C, \sD_\bE + D_I \otimes \gamma_5)$ represents the same Kasparov class as $({}_BE\otimes_A L^2(\bS)_\C, \sD_\bE)$, so it also represents this Kasparov product.
\end{proof}

\begin{remark}
\begin{enumerate}
 \item The construction of $I\times_\nabla M$ via Kasparov products fits naturally in the framework of Mesland's category of spectral triples \cite{Mes14}, where the internal space $I^\infty $ with the connection $\nabla$ can be seen as (a representative of) a morphism from the canonical triple for $M$ to the almost-commutative manifold $I^\infty\times_\nabla M$. 
 \item As is clear from the above discussion, the presence of the operator $D_I$ (or $D_I \otimes \gamma_5$) is completely irrelevant on the level of $KK$-classes. In this sense the $KK$-equivalence is too strong for our purposes, because in the models under consideration the presence of the operator $D_I$ certainly does matter. We will describe in \cref{sec:gauge_theory} how this operator plays the role of a `mass matrix' for the elementary fermions of the gauge theory, and gives rise to the Higgs field in the noncommutative Standard Model (see also \cref{sec:ex_ED} for a concrete example of $D_I$ as a mass matrix). 
\end{enumerate}
\end{remark}

\section{Principal modules}
\label{sec:gauge_modules}

We would like to describe a classical gauge theory on a manifold $M$ by considering an almost-commutative manifold $I^\infty\times_\nabla M$. For this purpose we now restrict our attention to a special case of internal spaces, which we call principal modules.

In \cref{ssct:classificationrealfinite} we first recall (part of) the classification of finite-dimensional real spectral triples that  has been done by Krajewski \cite{Kra98} and by Paschke and Sitarz \cite{PS98}. 
In \cref{ssct:principalmodules} we then define the notion of principal modules, and we show that, when the base manifold (which is of arbitrary dimension) is simply connected, the gauge group of a principal module (as defined for internal spaces in \cref{eq:gaugegroup_IS}) is isomorphic to the classical notion of the gauge group of a principal fibre bundle (as defined in \cref{defn:gaugegroup}). 

\subsection{Real finite spectral triples}
\label{ssct:classificationrealfinite}

Finite-dimensional real spectral triples have been classified 
for the case of KO-dimension $0$ \cite{Kra98,PS98}. With similar arguments, this can be generalised to arbitrary KO-dimension \cite{Sui14}. In the following theorem we give the result for complex algebras, while also setting the matrix $D_F = 0$. Below $\cc$ denotes complex conjugation of the coefficients with respect to the standard basis of $\mathbb{C}^{m_{ij}}$. 
\begin{thm}
\label{thm:formfinspec}
 Let $F := (A_F, \mathcal{H}_F,0,J_F)$ be a real finite spectral triple over a complex $*$-algebra $A_F$. Up to unitary equivalence, this triple is of the form
\begin{align*}
 A_F &= \bigoplus_{i=1}^l M_{N_i}(\mathbb{C}) , & \mathcal{H}_F &= \bigoplus_{i,j=1}^l \mathcal{H}_{ij} , & \mH_{ij} &:= \bigoplus_{i,j=1}^l M_{N_i,N_j}(\mathbb{C}) \otimes \mathbb{C}^{m_{ij}} , 
\end{align*} 
such that $m_{ij} = m_{ji}$, and the inner product on each copy of $M_{N_i,N_j}(\mathbb{C})$ is given by $\langle t_1, t_2\rangle = \Tr(t_1^* t_2 )$.
If $J_F^2 = \varepsilon$, then $J_F$ acts on $\mathcal{H}_{ij} \oplus \mathcal{H}_{ji}$, $(i<j)$, as 
\begin{align*}
 \left( \begin{array}{cc} 0 & \varepsilon (\cdot)^* \\ (\cdot)^* & 0   \end{array} \right) \otimes (\Id_{m_{ij}} \circ \cc).
\end{align*}
If $J_F^2 = 1$, the real structure $J_F$ acts on $\mathcal{H}_{ii} \simeq M_{N_i}(\mathbb{C}) \otimes \mathbb{C}^{m_{ii}}$ as  
\begin{align*}
(\cdot)^* \otimes (\Id_{m_{ii}}\circ \cc),
\end{align*}
If $J_F^2 = -1$, then $m_{ii}$ is even and $J_F$ acts on $(M_{N_i}(\mathbb{C}) \oplus M_{N_i}(\mathbb{C})) \otimes \mathbb{C}^{\frac{m_{ii}}{2}}$ as
\begin{align*}
\left(  \begin{array}{cc} 0 & - (\cdot)^* \\ (\cdot)^* & 0 \end{array} \right) \otimes ( \Id_{\frac{m_{ii}}{2}}\circ \cc). 
\end{align*}
The different copies of $M_{N_i,N_j}(\mathbb{C})$ (with respect to the above decomposition) in $\mathcal{H}_{ij}$ are denoted by $\mathcal{H}_{ij}^\alpha$, where $1 \leq\alpha \leq m_{ij}$.
\end{thm}

\begin{remark}
\label{rmk:formfinspec}
For finite-dimensional complex vector spaces $V$ and $W$, consider the linear isomorphism 
\begin{align*}
L\colon V \otimes \bar{W} \rightarrow \Hom(W,V), \quad v \otimes \bar{w} \mapsto (w' \mapsto v\langle w,w' \rangle), \quad v \in V,\; w,w' \in W,   
\end{align*}
where $\bar W$ denotes the conjugate vector space. 
Write $V_i = \mathbb{C}^{N_i}$, endowed with the standard inner product. Then the finite-dimensional Hilbert space $\mathcal{H}_{ij}$ can also be put in the form
\begin{align*}
 \mathcal{H}_{F} = \bigoplus_{(i,j) \in K} V_i \otimes \bar{V}_j,
\end{align*}
endowed with its standard inner product. Here $K$ is a \emph{multiset} consisting of pairs in $I \times I$ such that the multiplicity of $(i,j)$ is equal to $(j,i)$ and such that the projection $K\rightarrow I$ on either of the factors is surjective (this last condition is equivalent to the faithfulness of the action of $A_F$ on $\mathcal{H}_F$). The algebra $A_F \otimes A_F^{\op}$ acts on a summand $V_i \otimes \bar{V_j}$ as
\begin{align*}
 (a,b^{\op})(v \otimes \bar{w}) = a_iv \otimes \bar{b_j^*w},
\end{align*}
and the corresponding real structure on $V_i \otimes \bar{V_j} \rightarrow V_j \otimes \bar{V_i}$ is simply given by
\begin{align*}
 J_F(v \otimes \bar{w}) =  \pm w \otimes \bar{v},
\end{align*}
where the signs are determined by the $KO$-dimension of $F$. We will use this form of the real finite spectral triple in \cref{sct:gaugemodules}.
\end{remark}

From now on we assume  that every real finite spectral triple (with $D_F =0 $) is of the form as mentioned in \cref{thm:formfinspec}. Later on, the algebra $(A_F)_{J_F}$ will also be of interest, so we conclude this subsection by determining its precise form.

Recall that, in general, for any real spectral triple $(\A,\mH,D,J)$, the complex central subalgebra $\A_J$ is defined as $\A_J = \{ a \in \A \mid aJ = Ja^* \}$. 

\begin{prop}
\label{prop:AJ}
With notation as above, 
we have
$$
(A_F)_{J_F} = \Big\{a = \bigoplus_{i\in I}\lambda_i\id_{N_i} \in A_F \mathrel{\big|} \lambda_i\in\C;\; \lambda_i = \lambda_j \textnormal{ if } \mathcal{H}_{ij} \neq \{0 \}\Big\} .
$$
\end{prop}
\begin{proof}
We can assume that $J$ is in standard form. Write $A_F = \bigoplus_i M_{N_i}(\mathbb{C})$ and consider an element $a = \bigoplus_{i\in I} a_i \in A_F$. If $t \in \mathcal{H}_{ij}^\alpha$ ($1 \leq\alpha \leq m_{ij}$), then
\begin{align*}
 a (J_Ft) =   \pm a_j t^* \quad \text{and} \quad J_F(a^*t) = \pm t^* a_i.
\end{align*}
Choose $t^* = e_{kl}$, where $1 \leq k  \leq N_j$ and $1 \leq l \leq N_i$. Then 
\begin{align*} 
 (a_j e_{kl})_{\gamma \beta} = (a_j)_{\gamma k}   \delta_{\beta l}, \quad \text{and} \quad (e_{kl} a_i)_{\gamma \beta} = \delta_{\gamma  k} (a_i)_{l \beta},
\end{align*} 
Therefore, $aJ_F = J_Fa^*$ if and only if 
\begin{align*}
 (a_j)_{\gamma  k}   \delta_{\beta l} =(a_i)_{l \beta}  \delta_{\gamma  k} ,
\end{align*}
for all  $1 \leq k, \gamma  \leq N_j$ and $1 \leq l, \beta \leq N_i$. It follows that $a_i$, $a_j$ are diagonal and $(a_j)_{kk} = (a_i)_{ll}$ for all $1 \leq k\leq N_j$ and $1 \leq l \leq N_i$. Hence, $a \in (A_F)_{J_F}$ if and only if each $a_i = \lambda_i \text{id}_{N_i}$ and $\lambda_i = \lambda_j$ if $\mathcal{H}_{ij} \neq \{0\}$.   
\end{proof}

The following definition is inspired by the proof of \cref{prop:AJ}.
\begin{defn}
\label{defn:connected_irreps}
Let $A_F = \bigoplus_{i\in I} M_{N_i}(\C)$ act on $\mH_F = \bigoplus_{i,j\in I} \mH_{ij}$ as above. We define an equivalence relation on $I$ as follows. For $i\neq j\in I$ we set $i\sim j$ if there exists a sequence $i=i_0,\ldots,i_k=j$ such that $\mH_{i_m,i_{m+1}}\neq\{0\}$ for all $0\leq m<k$. If $i\sim j$ we say that \emph{$i$ is connected to $j$}. 
\end{defn}

\cref{prop:AJ} in particular shows that $\mathbb{C} \subset  (A_F)_{J_F}  \subset Z(A_F)$. 

\begin{coro}
\label{coro:AJ}
We have the isomorphism $(A_F)_{J_F} \simeq \bigoplus_{[i]\in I/\sim} \C$. In particular, the two extreme cases are:
\begin{itemize}
\item $(A_F)_{J_F} = Z(A_F)$ if and only if $\mathcal{H}_{ij} =0$ for all $i \neq j$ (that is, $I/\!\!\sim\;\simeq I$).
\item $(A_F)_{J_F} = \C$ if and only if $i$ is connected to $j$ for all $i,j\in I$ (that is, $I/\!\!\sim\;\simeq \{1\}$). 
\end{itemize} 
\end{coro}

\subsection{Principal modules}
\label{ssct:principalmodules}
We now want to find spectral triples for gauge theories that are globally non-trivial. Recall from \cref{defn:gauge_theory} that a general gauge theory with structure group $G_F$ on a manifold $M$ is given by a principal $G_F$-bundle $\bP$ over $M$ (along with a prescribed action functional or Lagrangian).

If $(A_F, \mathcal{H}_F, D_F, J_F)$ is a finite-dimensional real spectral triple, then the corresponding gauge group $G_F$ is given by (see also \cref{eq:gaugegroup_ST})
$$
G_F := \{uJ_FuJ_F^* \mid u\in\mU(A_F)\} \simeq \mathcal{U}(A_F) / \mathcal{U}((A_F)_{J_F}) .
$$
Such finite spectral triples can be used to describe globally trivial gauge theories over $M$ (see the Introduction). Any finite spectral triple $F$ automatically yields an internal space
$$
I^\infty_F = \big( \Gamma^\infty(M\times A_F) , \Gamma^\infty(M\times\mH_F), D_F , J_F \big) , 
$$
where now $D_F$ and $J_F$ are seen as constant bundle endomorphisms acting on the fibre $\mH_F$. 
We now want to generalise this construction in order to describe globally non-trivial gauge theories. Of course, fibre-wise we want to obtain the finite-dimensional situation that 
has been explained in \cref{ssct:classificationrealfinite}.  

The most straightforward way to obtain (examples of) globally non-trivial gauge theories over $M$ would then be as follows (see also \refcite[Lemma 2.5]{Cac12} and \refcite{BvS11}). Take any real finite spectral triple $F := (A_F, \mH_F, D_F, J_F)$ 
with gauge group $G_F$, 
and let $M$ be a smooth compact $4$-dimensional Riemannian spin manifold. Take any principal $G_F$-bundle $\bP\to M$. We construct the globally non-trivial triple of the form
\begin{align*}
\bP\times_{G_F}F := \big( \Gamma^\infty(\bP \times_{G_F} A_F) , \Gamma^\infty(\bP \times_{G_F} \mH_F) , D_\bP , 1\times J_F \big) .
\end{align*}
Here $D_\bP$ is an endomorphism acting on the vector bundle $\bP\times_{G_F}\mH_F$ satisfying certain compatibility requirements (which we will specify later in \cref{defn:mass_matrix}). 

\begin{remark}
Note that (in contrast to \refcite{Cac12}) we do not require $D_\bP$ to be of the form $1\times D_F$, where $D_F$ is a $G_F$-invariant operator on $\mH_F$, as such an assumption is too strong for our purposes. In particular, in specific examples (such as the noncommutative Standard Model) that requirement would prevent the appearance of a scalar (Higgs-like) field through inner fluctuations. 
\end{remark}

For the remainder of this section we ignore the endomorphism $D_\bP$, since it is not relevant for the definition of the gauge group, and we define the following:

\begin{defn}
Let $F := (A_F, \mH_F, 0, J_F)$ be a real finite spectral triple of the same form as in \cref{thm:formfinspec}. Write $G_F$ for the corresponding gauge group. Let $M$ be a smooth compact Riemannian spin manifold and let $\bP\to M$ be any principal $G_F$-bundle. A triplet of the form
\begin{align*}
\bP\times_{G_F}F := \big( \Gamma^\infty(\bP \times_{G_F} A_F) , \Gamma^\infty(\bP \times_{G_F} \mH_F) , 1\times J_F \big) ,
\end{align*}
is called a \emph{principal $G_F$-module} over $M$ (or $C^\infty(M)$) with fibre $F$. For brevity, we introduce the notation $\bB := \bP \times_{G_F} A_F$, $\bE := \bP \times_{G_F} \mH_F$, $\B := \Gamma^\infty(\bB)$, $\E := \Gamma^\infty(\bE)$, and $J := 1\times J_F$. 
\end{defn}

\begin{remark}
\label{rem:bundle-atlas}
 The principal fibre bundle $\bP$ is an explicit ingredient in the definition of a principal module. 
From $\bP$ we constructed the associated vector bundle $\bE = \bP\times_{G_F}\mH_F$, and (as discussed in \cref{sec:structure_group}) $\bP$ equips $\bE$ with a unique equivalence class of $G_F$-atlases. Whenever we consider transition functions of $\bE$, we therefore assume that they form a $G_F$-atlas in the equivalence class obtained from $\bP$.  
Given a $G_F$-atlas, the vector bundle $\bE$ inherits a hermitian structure from the inner product on $\mH_F$, which is well-defined because the action of $G_F$ on $\mH_F$ is unitary. For two equivalent $G_F$-atlases, the corresponding hermitian structures are isometric. 

We stress that, given only the vector bundle $\bE$ (with structure group $G_F$), we cannot reconstruct the principal $G_F$-bundle $\bP$. In order to reconstruct $\bP$, we also need to know the corresponding equivalence class of $G_F$-atlases. 
\end{remark}

\begin{prop}
A principal module $\bP\times_{G_F}F$ is a real internal space $(\Gamma^\infty(\bP \times_{G_F} A_F) , \Gamma^\infty(\bP \times_{G_F} \mH_F) , 0 , 1\times J_F)$ over $M$. 
\end{prop}
\begin{proof}
The action of $G_F$ on $A_F$ is given by conjugation when $A_F$ is considered as a $*$-subalgebra of $\text{End}(\mathcal{H}_F)$. Consequently, the fibre-wise action of the $*$-algebra bundle $\bB = \bP \times_{G_F} A_F$ on $\bE$ is well defined, and hence $\bB$ is a unital $*$-algebra subbundle of $\End(\bE)$. 
The operator $D_I = 0$ is trivially a hermitian endomorphism. 
Since the operator $J_F$ commutes with $G_F$, 
it induces a real structure $J_x$ on each fibre of $\bE$. The operator $J = 1 \times J_F$ denotes the anti-linear operator on $\bE$ that is induced by these real structures $J_x$ on the fibres.
\end{proof}

\begin{remark}
\label{rmk:innerautomorphisms}
Because $(uJ_FuJ_F^*) a (J_F u^*J_F^* u^*) 
= uau^*$ for all $a \in A_F$, $u \in  \mathcal{U}(A_F)$, we see that the given action of an element $uJ_FuJ_F^* \in G_F$ on $A_F$ coincides with the usual conjugation of the element $u \in \mU(A_F)$. Since $(A_F)_{J_F} \subset Z(A_F)$, the map $\tau\colon G_F \ni uJ_FuJ_F^* \mapsto \text{Ad}(uJ_FuJ_F^*) = \text{Ad }u \in \text{Inn}(A_F)$ does not depend on the choice of $u$. 
Thus, the surjective map $\tau\colon G_F \rightarrow \mathcal{U}(A_F) / \mathcal{U}(Z(A_F)) \simeq \text{Inn}(A_F)$ is induced by the usual map $\mathcal{U}(A_F) \rightarrow \text{Inn}(A_F)$ (recall that $G_F$ is the quotient $\mathcal{U}(A_F) / \mathcal{U}((A_F)_{J_F})$).
\end{remark}

\subsubsection{The gauge group}
\label{sssct:gaugegroup}
Consider a principal module $\bP\times_{G_F}F =\big( \B, \E, J \big)$ over $M$. 
Using the classification of $A_F$ and $\mH_F$, as given in \cref{ssct:classificationrealfinite}, we can decompose the bundles $\bB = \bP\times_{G_F}A_F$ and $\bE = \bP\times_{G_F}\mH_F$ in a similar way:
\begin{align*}
\bB &= \bigoplus_{i\in I} \bB_i , & \bB_i &= \bP\times_{G_F}M_{N_i}(\C) , \\
\bE &= \bigoplus_{i,j\in I} \bE_{ij} , & \bE_{ij} &= \bP\times_{G_F}\mH_{ij} . 
\end{align*}
Each vector bundle $\bE_{ij}$ carries the obvious action by $\bB\otimes\bB^\op$. Note, however, that even though $\mH_{ij} = \C^{N_i}\otimes\C^{N_j}\otimes\C^{m_{ij}}$, $\bE_{ij}$ is not necessarily of the form $\bE_i\otimes \bE_j\otimes \C^{m_{ij}}$ for some vector bundles $\bE_i$ and $\bE_j$ (see \cref{sec:ex_YM} for an example).

Note that, for the case $i=j$, the bundle $\bE_{ii}$ is necessarily isomorphic to (a number of copies of) $\bB_i$. Indeed, the $G_F$-valued transition functions act on the fibres of $\bE_{ii}$, which are isomorphic to (copies of) $M_{N_i}(\mathbb{C})$, by conjugation with an element $u \in U(N_i)$, and are therefore inner automorphisms of the algebra $M_{N_i}(\mathbb{C})$. By \cref{rmk:innerautomorphisms} these transition functions are equal to those for the $*$-algebra bundle $\bB_i$.

Denote by $[i]$ the equivalence class of all $j\in I$ that are connected to $i$ (see \cref{defn:connected_irreps}). Write 
\begin{align*}
\bB_{[i]} = \bigoplus_{s\in[i]}  \bB_s,
\end{align*}
and write $b_{[i]}$ for the projection of an element $b$ onto $\bB_{[i]}$. 
As $(\B_{[i]})_J = C^\infty(M)$ we obtain (see also \cref{coro:AJ})
$$
\B_J = \bigoplus_{[i]\in I/\sim} C^\infty(M).
$$

The gauge group of the principal module $\bP\times_{G_F}F = (\B,\E,J)$ is defined as (see \cref{eq:gaugegroup_IS})
$$
\G(\bP\times_{G_F}F) := \big\{ uJuJ^* \mid u\in\mU(\B) \big\} \simeq \mU(\B) / \mU(\B_J).
$$

At the same time, a principal $G_F$-bundle $\bP\to M$ is equipped with the gauge group $\G(\bP) = \Gamma^\infty(\Ad\bP)$ (see \cref{sec:pfb}). We now aim at showing that for a principal module $\bP\times_{G_F}F$, the gauge groups $\G(\bP\times_{G_F}F)$ and $\G(\bP)$ coincide, provided that $M$ is simply connected.

Consider the group bundle map
\begin{align*}
\phi&\colon  \mathcal{U}(\bB) \simeq \bP \times_{G_F} \mathcal{U}(A_F) \rightarrow \bP \times_{G_F} \mathcal{U}(\mathcal{H}_F), &  u_x &\mapsto u_xJ_xu_xJ_x^*.
\end{align*}
The image $\phi(\mU(\bB))$ is a group subbundle of $\bP \times_{G_F} \mathcal{U}(\mathcal{H}_F)$, with fibres isomorphic to $G_F$. In fact, this subbundle is isomorphic to the group bundle $\Ad\bP$. The induced map $\phi_*$ on the sections $\mathcal{U}(\mathcal{B}) \simeq \mathcal{U}(\Gamma^\infty( \bB )) \rightarrow \mathcal{U}(\Gamma^\infty(\End(\bE)) )$ is precisely the map $u \mapsto uJuJ^*$, $u \in \mathcal{U}(\mathcal{B})$. Thus, $\phi_*$ maps $\mathcal{U}(\mB)$ into $\Gamma^\infty(\Ad\bP)$. 
However, as discussed in \cref{sec:covering}, this map $\phi_*$ need not always be surjective. We will proceed by showing that in our case, under the assumption that $M$ is simply connected, we do have surjectivity. 

\begin{prop}
\label{prop:covering}
 Let $\bP\times_{G_F}F$ be a principal module over $M$. There exists a group subbundle $\bU \subset \mathcal{U}(\bB)$ such that the restriction $\phi\colon \bU \rightarrow \Ad\bP$ is a covering map.
\end{prop}
\begin{proof}
Consider the subbundle $\bE_{[i]} := \bB_{[i]} \cdot \bE$ (i.e.\ the subbundle on which $\bB_{[i]}$ acts non-trivially). 
Define the group subbundle
\begin{align*}
 \bU := \{ u \in \mathcal{U}(\bB) \mid \det\!{}_{[i]} u_{[i]} = 1 \text{ for all } [i] \}, 
\end{align*}
where $\det_{[i]} u_{[i]}$ denotes the fibrewise determinant of $u_{[i]}$ seen as an element of the bundle $\End \bE_{[i]}$. 
Denote the rank of $\bE_{[i]}$ by $N_{[i]}$. 
Since any element $u \in \mathcal{U}(\bB)$ can be written as $u=vw$, where $v  \in\bU$ and $w  \in \mathcal{U}(\bB_J)$ (just take $w_{[i]} = \left(\det_{[i]} u_{[i]}\right)^{\frac{1}{N_{[i]}}} \text{id}_{N_{[i]}}$ and $v = uw^{-1}$), the image $\phi(\bU)$ is equal to the image $\phi(\mathcal{U}(\bB)) = \Ad\bP$. 

Let us calculate the kernel $\phi_x\colon \bU_x \rightarrow (\Ad\bP)_x$. Choose $u \in \bU_x \cap \ker \phi_x$. Since $u \in \ker \phi_x$, each $u_{[i]}$ is diagonal. Because $\det_{[i]} u_{[i]} =1$, we obtain that $u_{[i]} = \lambda_{[i]} \text{id}_{N_{[i]}}$, where $\lambda_{[i]}$ is an $N_{[i]}$-th root of unity. Since there are only finitely many equivalence classes $[i]$, the group $\bU_x \cap \ker \phi_x$ is finite. 

The condition for a map to be a covering map is of a local nature, so we can assume that all bundles are globally trivial. In that case, it follow from the fact that $\bU_x \cap \ker \phi_x$ is finite, that $\bU \rightarrow \Ad\bP$ is a covering map.
\end{proof}

Combining \cref{prop:covering} with \cref{coro:surjective,thm:1} immediately yields the desired result:

\begin{thm}
\label{coro:gaugegroupsisomorphic}
 Let $\bP\times_{G_F}F$ be a principal module over $M$. If $M$ is simply connected, then 
 $$
 \G(\bP\times_{G_F}F) \simeq \Gamma^\infty(\Ad\bP) \simeq \G(\bP).
 $$
\end{thm}

\begin{remark}
\label{remark:unimod}
It follows from the above that for each element $g$ of the gauge group $\G(\bP\times_{G_F}F)$, there exists a unitary section $u\in \B$ with (fibre-wise) determinant equal to $1$, such that $g = uJuJ^*$. 
In this sense, the gauge group is \emph{unimodular} by default. This only holds for complex algebras $\B$. For real algebras (including the one describing the noncommutative Standard Model \cite{Connes96,CCM07}) one needs to impose unimodularity by hand (see also \refcite{LS01} and references therein). 
\end{remark}

\section{Gauge modules}
\label{sct:gaugemodules}
In \cref{ssct:principalmodules} we introduced the notion of principal modules, which have an entirely \emph{geometric} nature. In this section we introduce so-called \emph{gauge modules}, which are of a purely \emph{algebraic} nature. We show that each gauge module is in fact also a principal module, but unfortunately not all principal modules can be obtained from gauge modules. 

Inspired by the standard form of finite spectral triples as obtained in \cref{thm:formfinspec,rmk:formfinspec}, we introduce the following definition, which might be considered an extension of Krajewski diagrams to the globally non-trivial case.
\begin{defn}
\label{defn:gaugemodule}
 Let $\A := C^\infty(M)$. Suppose we are given a finite set of non-degenerate hermitian finitely generated projective $\A$-modules $\E_i$ (for $i\in I = \{1, \dots, l\}$), and define the module algebras $\B_i := \End_\A(\E_i)$. Take a  \emph{multiset} $K$ consisting of pairs in $I\times I$ such that the multiplicity of $(i,j)$ is equal to the multiplicity of $(j,i)$, and such that the projection $K\to I$ on either of the factors is surjective. Denote the multiplicity of the pair $(i,j)$ by $m_{ij}$ and write $(i_\alpha,j_\alpha)$ ($1 \leq \alpha\leq m_{ij}$) to distinguish the pairs in $K$ that occur more than once (see also \cref{thm:formfinspec} for this notation). 

A \emph{gauge module} $(\B,\E,J)$ is of the form
\begin{align*}
\B &:= \bigoplus_{i\in I} \B_i , & \E &:= \bigoplus_{(i,j)\in K} \E_i\otimes_\A \bar{\E_j} , & J&\colon \E_i\otimes_\A\bar{\E_j}\to\E_j\otimes_\A\bar{\E_i},
\end{align*}
where $J$ is of the same standard form as the finite operator $J_F$ in \cref{thm:formfinspec} (and which depends on the value of $J^2 = \varepsilon = \pm 1$, e.g.\ $J_{ij}(e_{i_\alpha} \otimes \bar{e_{j_\alpha}}) = \varepsilon e_{j_\alpha} \otimes \bar{e_{i_\alpha}}$, for $e_{i_\alpha} \otimes \bar{e_{j_\alpha}} \in \E_{i_\alpha} \otimes \bar{\E_{j_\alpha}}$ if $j<i$). 
\end{defn}

The assumption that the projection $K\to I$ is surjective ensures that the action of $\B$ on $\E$ is faithful. 
From the Serre-Swan theorem we know that each module $\E_i$ is given by the smooth sections of a vector bundle $\bE_i\to M$. Because the hermitian structure on $\E_i$ is non-degenerate, this yields a hermitian structure on $\bE_i$. 
By \cref{thm:serreswanalgebrabundles} the module algebra $\B_i$ is given by the smooth sections of a unital weak $*$-algebra bundle $\bB_i\to M$. Since $\B_i = \End_\A(\E_i)$ we obtain $\bB_i = \End(\bE_i)$. The local triviality of $\bB_i$ then follows from the local triviality of $\bE_i$, which means that $\bB_i$ is in fact a unital $*$-algebra bundle. 

As mentioned in \cref{rem:bundle-atlas}, given a principal module $\bP\times_{G_F}F = (\B,\E,J)$ (but not $\bP$ itself), it is not possible to reconstruct $\bP$, unless we are given the equivalence class of $G$-atlases on the vector bundle $\bE = \bP\times_{G_F}\mH_F$. 
However, we will show below that for gauge modules it is possible to uniquely reconstruct the corresponding principal $G_F$-bundle. The main distinctive feature of gauge modules is that the vector bundle $\bE$ decomposes as a direct sum of tensor products of hermitian vector bundles $\bE_i$. To each $\bE_i$ there uniquely (up to isomorphism) corresponds a principal $U(N_i)$-bundle. From these principal $U(N_i)$-bundles we can subsequently construct the corresponding principal $G_F$-bundle $\bP$. 

\begin{prop}
\label{prop:gaugemodulestructuregroup}
Let $(\B,\E,J)$ be a gauge module. Then:
\begin{enumerate}
 \item There exist a real finite spectral triple $F = (A_F, \mH_F, 0, J_F)$ and a principal $\mU(A_F)$-bundle $\bQ$ such that 
$(\B,\E,J) = Q \times_{\mathcal{U}(A_F)} F$.
 \item There exists a principal $G_F$-bundle $\bP$ such that $(\B, \E, J) = \bP\times_{G_F}F$. 
\end{enumerate}
\end{prop}
\begin{proof}
\begin{enumerate}
\item 
The gauge module $(\B,\E,J)$ is constructed from a given set of hermitian vector bundles $\bE_i$ of rank $N_i$ and the index (multi)sets $I$ and $K$. By assumption $\bB_i = \End(\bE_i)$, and so $\bB_i$ has typical fibre $M_{N_i}(\C)$. We define
\begin{align*}
A_F &:= \bigoplus_{i\in I} M_{N_i}(\C) , & \mH_F &:= \bigoplus_{(i,j)\in K} \C^{N_i}\otimes\bar{\C^{N_j}} . 
\end{align*}
For each $\bE_i$, there is a principal $U(N_i)$-bundle $\bQ_i$ (which is unique up to isomorphism) such that $\bE_i \simeq \bQ_i\times_{U(N_i)}\C^{N_i}$ (see \cref{ex:unitary-princ}). 
Let $(U,u\Sub{UV}^i)$ be a $U(N_i)$-atlas on $\bE_i$ corresponding to local trivialisations of $\bQ_i$. 
The transition functions $\bar{u\Sub{UV}^j}$ of $\bar{\bE_j}$ are given by the right action of $(u\Sub{UV}^j)^*$ on $\bar{\C^{N_j}}$ (see \cref{eq:transitionfunctionsconjugatevectorbundle}), which is implemented as $(v_i\otimes\bar{w_j})(u\Sub{UV}^j)^* = Ju\Sub{UV}^jJ^*(v_i\otimes\bar{w_j})$. Hence we obtain transition functions for $\bE$ of the form
$$
g\Sub{UV} = \bigoplus_{(i,j)\in K} u\Sub{UV}^i \otimes {(u\Sub{UV}^j)^*}^\op = \bigoplus_{(i,j)\in K} u\Sub{UV}^iJu\Sub{UV}^jJ^* .
$$
Writing $u\Sub{UV} = \bigoplus_{i\in I} u\Sub{UV}^i \in C^\infty(U\cap V,\mU(A_F))$, we see that $g\Sub{UV} = u\Sub{UV}Ju\Sub{UV}J^* \in C^\infty(U\cap V,G_F)$. Since the $u\Sub{UV}^i$ are transition functions of $\pi_i\colon\bQ_i\to M$, we see that the $u\Sub{UV}$ are the transition functions of the principal $\mU(A_F)$-bundle 
$$
\bQ := \bQ_1\times_M\cdots\times_M\bQ_{l} := \{ (q_1, \dots,q_{l}) \in \bQ_1 \times \cdots \times \bQ_{l} \mid \pi_1(q_1) = \cdots = \pi_l(q_{l}) \} . 
$$
Since the action of $u\Sub{UV}$ on $\mathcal{H}_F$ is given by $g\Sub{UV} = u\Sub{UV}Ju\Sub{UV}J^*$, we see that $\bE \simeq \bQ\times_{\mU(A_F)}\mH_F$ as hermitian vector bundles. As conjugation by $u\Sub{UV}$ coincides with conjugation by $g\Sub{UV}$ on the algebra $A_F$, we also have $\bB \simeq \bQ\times_{\mU(A_F)}A_F$. It is straightforward to check that $J$ is invariant under conjugation by a transition function $g\Sub{UV}$, and hence it is simply of the form $J = 1\times J_F$. Since $J$ is an anti-unitary operator satisfying $J^2=\varepsilon$ and the order-zero condition, it follows that $J_F$ is a real structure on $\mH_F$. 
\item 
Given the principal $\mU(A_F)$-bundle $\bQ$ from the first part of this lemma, we simply construct a principal $G_F$-bundle as
\begin{align*}
\bP &:= \bQ \times_{\mU(A_F)}G_F,
\end{align*}
where $u \in \mU(A_F)$ acts on $G_F$ as left multiplication by the element $uJ_FuJ_F^*$. 
The transition functions of $\bP$ are given by $g\Sub{UV} = u\Sub{UV}Ju\Sub{UV}J^* \in C^\infty(U\cap V,G_F)$. 
It then straightforwardly follows that 
\begin{align*}
\bP\times_{G_F}\mH_F &\simeq (\bQ \times_{\mU(A_F)}G_F)\times_{G_F}\mH_F \simeq \bQ\times_{\mU(A_F)}\mH_F \simeq \bE , 
\end{align*}
and similarly we obtain $\bP\times_{G_F}A_F \simeq \bB$.  \qedhere
\end{enumerate}
\end{proof}

The above proposition shows that each gauge module is in fact a principal module $\bP\times_{G_F}F$ (where we can uniquely reconstruct $F$ and $\bP$), where $\bP$ can be lifted to a principal $\mU(A_F)$-bundle $\bQ$ (which is unique up to isomorphism). We now show the converse, namely that a principal module $\bP\times_{G_F}F$ with a lift $\tau\colon\bQ\to\bP$ uniquely corresponds to a gauge module. 

\begin{prop}
\label{prop:principal_with_lift}
Let $\bP\times_{G_F}F = (\B,\E,J)$ be a principal module, and suppose we have a principal $\mU(A_F)$-bundle $\bQ$ that lifts $\bP$. Then $\bQ$ naturally induces a gauge module structure on $(\B,\E,J)$. 
\end{prop}
\begin{proof}
As we have seen in \cref{ssct:classificationrealfinite}, the real finite spectral triple $F = (A_F, \mH_F, 0, J_F)$ has a decomposition of the form
\begin{align*}
 A_F &= \bigoplus_{i\in I} M_{N_i}(\C) , & \mH_F &= \bigoplus_{(i,j)\in K} \C^{N_i}\otimes\bar{\C^{N_j}} . 
\end{align*}
Thus we have $\mU(A_F) = \times_{i\in I} U(N_i)$, and the principal $\mU(A_F)$-bundle $\bQ$ then decomposes as $\bQ_1 \times_M \cdots \times_M \bQ_{l}$, where each $\bQ_i$ is a principal $U(N_i)$-bundle given by $\bQ_i := \bQ \times_{\mU(A_F)} U(N_i)$. 
We then construct
\begin{align*}
 \bB_i &:= \bQ \times_{\mU(A_F)} M_{N_i}(\C) \simeq \bQ_i \times_{U(N_i)} M_{N_i}(\C) , & \bE_i &:= \bQ \times_{\mU(A_F)} \C^{N_i} \simeq \bQ_i \times_{U(N_i)} \C^{N_i} ,
\end{align*}
where $\mU(A_F) = \times_{i\in I} U(N_i)$ acts on $\C^{N_i}$ as left multiplication by the factor $U(N_i)$, and on $M_{N_i}(\C)$ as conjugation by $U(N_i)$. The bundle $\bE_i$ naturally inherits a hermitian structure from the standard inner product on $\C^{N_i}$. Because $\bQ$ lifts $\bP$, the bundles $\bB$ and $\bE$ corresponding to the principal module $\bP\times_{G_F}F$ are in fact of the form
\begin{align*}
 \bB &:= \bQ\times_{\mU(A_F)}A_F = \bigoplus_{i\in I} \bB_i , & \bE &:= \bQ\times_{\mU(A_F)}\mH_F = \bigoplus_{(i,j)\in K} \bE_i \otimes \bar{\bE_j} .
\end{align*}
Furthermore, as the transition functions of $\bB_i$ are given by conjugation by the transition functions of $\bE_i$, and as its fibre equals $M_{N_i}(\C) = \End(\C^{N_i})$, it follows that $\bB_i = \End(\bE_i)$ and $\bB_i$ acts as such on $\bE$. Hence we have shown that the principal module $\bP\times_{G_F}F$ is equal to the gauge module given by the modules $\E_i := \Gamma^\infty(\bE_i)$ and the real structure $J=1\times J_F$. 
\end{proof}

The previous two propositions then lead us to the main result of this section:
\begin{thm}
 A gauge module is characterised uniquely (up to isomorphism) by a principal module $\bP\times_{G_F}F$ for which there exists a principal $\mU(A_F)$-bundle $\bQ$ that lifts $\bP$. 
\end{thm}
\begin{proof}
Given a gauge module, we have shown in \cref{prop:gaugemodulestructuregroup} that we can uniquely construct a real finite spectral triple $F = (A_F, \mH_F, 0, J_F)$, a principal $G_F$-bundle $\bP$, and a principal $\mU(A_F)$-bundle $\bQ$ that lifts $\bP$. Conversely, given such $F$, $\bP$, and $\bQ$, \cref{prop:principal_with_lift} shows that $\bP\times_{G_F}F$ is in fact given by a gauge module. These constructions are inverse to each other.
\end{proof}

\begin{remark}
\begin{enumerate}
\item If there exists a principal $\mU(A_F)$-bundle $\bQ$ that lifts $\bP$, then $\bQ$ is unique up to isomorphism, because each principal $U(N_i)$-bundle $\bQ_i$ is unique up to isomorphism (cf.\ \cref{ex:unitary-princ}). 
\item Every \emph{globally trivial} principal module, constructed from a finite spectral triple $F$ and the principal bundle $\bP = M\times G_F$, is in fact a gauge module, with the lift $\bQ = M\times \mU(A_F)$. 
\item An example of a principal module that is (in general) \emph{not} a gauge module (except when for instance the underlying manifold is simply connected and $4$-dimensional) is described in \cref{sec:ex_YM}.  
\end{enumerate}
\end{remark}

\section{Gauge theory}
\label{sec:gauge_theory}

In this section we show how principal modules describe gauge theories on $4$-dimensional compact spin manifolds. First we will introduce a `mass matrix'. Viewing the (now massive) principal module as an internal space and endowing it with a (suitable) connection, we can then use it to construct an almost-commutative manifold. Subsequently, we determine the inner fluctuations and provide an explicit formula for the spectral action of this almost-commutative manifold. We end this section by stating our main result, namely that such an almost-commutative manifold indeed describes a gauge theory in the sense of \cref{defn:gauge_theory}. 

\subsection{Principal almost-commutative manifolds}
\label{sct:pacm}
\begin{defn}
\label{defn:mass_matrix}
Consider a principal module $\bP\times_{G_F}F = (\B,\E,J_I)$ (from here on we include a subscript $I$ in order to differentiate between the different operators occurring).
In order to be able to describe massive gauge theories, we now introduce a \emph{`mass matrix'}
$$
D_I\in\Gamma^\infty(\End(\bE))\simeq\End_{\A}(\E) ,
$$
satisfying
\begin{align*}
D_I &= D_I^* , & D_IJ_I &= \varepsilon'J_ID_I , & \big[[D_I,a],JbJ^*\big] &= 0 \quad\forall a,b\in\B ,
\end{align*}
where the sign $\varepsilon'$ (along with the signs $\varepsilon,\varepsilon''$ obtained through the finite spectral triple $F$) is determined by the KO-dimension according to the same table as in \cref{defn:spectral_triple}. 
We then call $I^\infty_\bP := (\B,\E,D_I,J_I)$ a \emph{massive principal module over $M$}. We say $I^\infty_\bP$ is \emph{even} if there exists a grading operator $\gamma_I$ on $\E$ such that $D_I\gamma_I = -\gamma_ID_I$, $\gamma_I J_I = \varepsilon'' J_I \gamma_I$ and $a\gamma_I = \gamma_I a$ for all $a \in \B$.  
\end{defn}
 It is an immediate consequence of the definition that a massive principal module over $M$ is a real internal space over $M$. If $(\B,\E,J_I)$ is in fact a gauge module, we shall call $(\B,\E,D_I,J_I)$ a \emph{massive gauge module}. 

Let $\bP\times_{G_F}F$ be a principal module. Denote by $\g_F$ the Lie algebra of the structure group $G_F$. Take a connection on $\bP$, i.e.\ for each local trivialisation $(U_i,h_i)$ of $\bP$ we have a (local) $\mathfrak{g}_F$-valued $1$-form $\omega_i \in \Omega^1(U_i,\g_F)$ such that 
\begin{align*}
 \omega_j = g_{ij}^{-1} dg_{ij} + g_{ij}^{-1} \omega_i g_{ij}
\end{align*}
for all $i,j$ such that $U_i\cap U_j\neq\emptyset$ (see \cref{defn:connection}). These connection one-forms yield a connection $\nabla\colon\E\to\E\otimes_\A\Omega^1(M)$ by defining locally (i.e.\ on local trivialisations $(U_i,h_i)$ of $\bE$ that are induced by those of $\bP$) the expression
$$
\nabla|_{U_i} := h_i^{-1} \circ (d + \omega_i) \circ h_i ,
$$
where $d$ is the exterior derivative acting on the components of the local trivialisation.
The transformation property of $\omega_i$ ensures that $\nabla$ is globally well-defined. Connections on $\E$ of this form are also referred to as \emph{$G_F$-compatible connections}, or simply $G_F$-connections. 

Consider the associated vector bundle $\ad\bP := \bP\times_{\ad}\g_F$, where 
$\ad$ is the adjoint action of $G_F$ on $\g_F$. Since $\g_F$ is (isomorphic to) the image of $\lu(A_F)$ in $\lu(\mH_F)$ under the map $t\mapsto t+J_FtJ_F^*$, the bundle $\ad \bP$ is (isomorphic to) the image of $\mathfrak{u}(\bB)$ in $\mathfrak{u}(\bE)$ under the map $\tau\colon t \mapsto t + J_ItJ_I^*$. The kernel of this map is equal to the set of all elements $t \in \mathfrak{u}(\bB)$ satisfying $t = -J_ItJ_I^* = J_It^*J_I^*$, or equivalently,
\begin{align*}
 \ker \tau = \{ t \in \mathfrak{u}(\bB) \mid tJ_I = J_It^* \} = \mathfrak{u}(\bB_J).
\end{align*}
Hence we see that $\ad\bP$ is isomophic to $\lu(\bB) / \lu(\bB_J)$. In particular, $\mathfrak{g}_F = \mathfrak{u}(A_F) / \mathfrak{u}((A_F)_{J_F})$. 

\begin{lem}
\label{lem:adP}
The induced map $\tau\colon\lu(\B)\to\Gamma^\infty(\ad\bP)$ is surjective, and 
$$
\Gamma^\infty(\ad\bP) \simeq \lu(\B)/\lu(\B_J) .
$$
Moreover, $\ad\bP$ is isomorphic to the subbundle 
\begin{align*}
\mathfrak{u} = \{ t \in \mathfrak{u}(\bB) \mid \Tr_{[i]} t_{[i]} = 0 \textnormal{ for all }[i] \}
\end{align*}
of $\lu(\B)$, where $\Tr_{[i]} t_{[i]}$ denotes the fibrewise trace of $t_{[i]}$ seen as an element of the bundle $\End \bE_{[i]}$, and
$ 
\lu(\bB) = \ker \tau \oplus \mathfrak{u},
$
with $\ker \tau = \mathfrak{u}(\bB_J)$.
\end{lem}
\begin{proof}
Though the first two statements follow immediately from the exactness of the Serre-Swan equivalence functor $\Gamma^\infty$,  we prove them directly by showing that $\ad\bP$ is isomorphic to the subbundle $\mathfrak{u}$ (compare also \cref{prop:covering}). Indeed, every $t \in \lu(\bB)$ can be written as $s + q$, where $s \in \mathfrak{u}$ and $q \in \mathfrak{u}(\bB_J)$ (just take $q_{[i]} =  \frac{1}{N_{[i]}} \Tr_{[i]}(t_{[i]})\cdot \text{id}_{[i]}$ and $s = t - q$). Hence $\tau|_\lu$ is surjective. 

Suppose now that $t \in \ker \tau|_\mathfrak{u}$. Because $t \in \ker \tau$, we obtain $t_{[i]} = \lambda_{[i]} \text{id}_{N_{[i]}}$, where $\lambda_{[i]} \in i\mathbb{R}$ (see \cref{prop:AJ}). Since $t \in \mathfrak{u}$, each of the $t_{[i]}$ is traceless. Hence each of the $\lambda_{[i]}$ is zero, and consequently, the kernel of $\tau|_\mathfrak{u}$ is trivial.
\end{proof}

\begin{lem}
\label{conn-comm}
Let  $\bP\times_{G_F}F = (\B,\E,J_I,\gamma_I)$ be an even principal module. 
Any $G_F$-compatible connection $\nabla$ on $\E$ commutes with the real structure $J_I$ (in the sense that $\nabla_\mu J_I = J_I \nabla_\mu$) and the grading $\gamma_I$. 
\end{lem}
\begin{proof}
It is sufficient to show that $J_F$ and $\gamma_F$ commute with elements in $\mathfrak{g}_F$. Any element in $\mathfrak{g}_F$ is of the form $t + J_F tJ_F^*$, with $t \in \mathfrak{u}(A_F)$. In particular, $J_F$ commutes with these elements. Since $\gamma_F$ commutes with elements in $A_F$, and (anti-)commutes with $J_F$, the grading $\gamma_F$ commutes with elements in $\mathfrak{g}_F$, too.
\end{proof}

If the principal module is obtained from a gauge module $(\B,\E,J_I)$, we can construct such a $G_F$-connection explicitly as follows. Consider the decomposition $\E = \bigoplus_{(i,j)} \E_i\otimes_\A\bar{\E_j}$, and choose a hermitian connection $\nabla^i$ on each $\E_i$. We define
\begin{align*}
\nabla &:= \bigoplus_{(i,j)} \left( \nabla^i\otimes\1 + \1\otimes\bar{\nabla^j} \right) ,
\end{align*}
where the conjugate connection $\bar{\nabla^j}$ is defined in \cref{sec:conjugate}. 
In order to see that $\nabla$ corresponds to a connection on the principal bundle $\bP$, we first need to check that its local connection one-forms take values in the Lie algebra $\g_F$. 
If $(U,h^i\Sub{U})$ are local trivialistions of $\E_i$, we can write $\nabla^i|_U = \left(h^i\Sub{U}\right)^{-1} \circ (d + \omega^i\Sub{U}) \circ h^i\Sub{U}$ for some local connection one-forms $\omega^i\Sub{U} \in \Omega^1(U,\lu(N_i))$. The connection $\nabla$ then locally has the connection $1$-form
\begin{align*}
\omega\Sub{U} &:= \bigoplus_{(i,j)} \left( \omega^i\Sub{U}\otimes\1 + \1\otimes({\omega^j\Sub{U}}^*)^\op \right) \in \Omega^1(U,A_F\otimes A_F^\op) .
\end{align*}
This ensures that $[\nabla,\cdot]$ yields a connection on $\B\otimes\B^\op$, which preserves $\B$ and $\B^\op$. 
Writing $t\Sub{U} = \bigoplus_{i\in I} \omega^i\Sub{U}$, we can write $\omega\Sub{U} = t\Sub{U} + J_Ft\Sub{U}J_F^* \in \Omega^1(U,\mathfrak{g}_F)$. To verify that $\omega\Sub{U}$ defines a connection on the principal $G_F$-bundle $\bP$ we need to show that $\omega\Sub{U}$ transforms correctly under the $G_F$-valued transition functions.

So, consider two neighbourhoods $U$ and $V$ such that $U\cap V\neq\emptyset$, and let $u = \times u_i \in C^\infty(U\cap V,\mU(A_F))$ be a transition function for the principal $\mU(A_F)$-bundle $\bQ$. The corresponding transition function for the principal $G_F$-bundle $\bP$ is $g:=uJ_FuJ_F^*$. Since the $\omega^i\Sub{U}$ are connection forms on $\bE_i$, $t\Sub{U}$ transforms as
$$
t\Sub{V} = \bigoplus_{i\in I} \omega^i\Sub{V} = \bigoplus_{i\in I} (u_i^* \omega^i\Sub{U} u_i + u_i^* du_i) = u^* t\Sub{U} u + u^* du .
$$
We then see that 
\begin{align*}
\omega\Sub{V} &= t\Sub{V} + J_F t\Sub{V} J_F^* = u^* t\Sub{U} u + u^* du + J_F (u^* t\Sub{U} u + u^* du) J_F^* \\
&= u^*J_Fu^*J_F^*  t\Sub{U} uJ_FuJ_F^* + J_F(u^*J_Fu^*J_F^*  t\Sub{U} uJ_FuJ_F^* ) J_F^* \\
&\quad+ u^*J_Fu^*J_F^* (du)J_FuJ_F^* +  u^*J_Fu^*J_F^* u J_F(du)J_F^* \\ 
&= g^{-1} (t\Sub{U} + J_Ft\Sub{U}J_F^*)g + g^{-1}dg = g^{-1} \omega\Sub{U} g +  g^{-1}dg .
\end{align*}
Thus, $U \mapsto \omega\Sub{U}$ indeed defines a $G_F$-connection.

\begin{prop}
\label{prop:connectiongaugemodule}
 Let $(\B,\E,J)$ be a gauge module. A connection on $\bE$ is of the form $\bigoplus_{(i,j)} \left( \nabla^i \otimes  \1 + \1 \otimes \bar{\nabla^j} \right)$ if and only if it induces a connection on the principal $\mathcal{U}(A_F)$-bundle $\mathtt{Q}$ from \cref{prop:gaugemodulestructuregroup}. 
\end{prop}
\begin{proof}
Consider a local trivialisation $(U,h\Sub{U})$ of $\bP$, and let $\omega\Sub{U}\in\Omega^1(U,\lu(A_F))$ be a local connection form on $\bQ$, yielding a connection $\nabla$ on $\bE = \bQ\times_{\mU(A_F)} \mH_F$. Since the decomposition $\lu(A_F) = \bigoplus_{i\in I} \lu(N_i)$ is preserved by the action of $\mathcal{U}(A_F) $, we can write $\omega\Sub{U} = \bigoplus_{i\in I} \omega_i$, where each $\omega_i\in\Omega^1(U,\lu(N_i))$ yields a connection $\nabla^i$ on $\bE_i$. For $x\in U$, the connection form $\omega\Sub{U}$ acts on $(\bE_i \otimes \bar{\bE_j})|_x \simeq \C^{N_i}\otimes \bar{\C^{N_j}}$ as
 $$
 \omega(v_i\otimes\bar{w_j}) = \omega_iv_i\otimes\bar{w_j} + v_i\otimes\bar{w_j}\omega_j^* ,
 $$
 from which it follows that $\nabla = \bigoplus_{(i,j)} \left( \nabla^i \otimes  \1 + \1 \otimes \bar{\nabla^j} \right)$.
 
 For the converse, consider a connection on $\bE$ of the form $\nabla = \bigoplus_{(i,j)} \left( \nabla^i \otimes  \1 + \1 \otimes \bar{\nabla^j} \right)$. On a local trivialisation $(U,h\Sub{U})_i$ of $\bE_i$, each connection $\nabla^i$ yields a local connection form $\omega_i \in\Omega^1(U,\lu(N_i))$. Then $\omega\Sub{U} := \bigoplus_{i\in I}\omega_i \in \Omega^1(U,\lu(A_F))$ is a connection form on $\bQ$ that induces $\nabla$. 
\end{proof}

\begin{defn}
Let $I^\infty_\bP = (\B,\E,D_I,J_I)$ be a massive principal module of KO-dimension $k$ over $M$, where $M$ now has dimension $4$. Let $\nabla$ be a $G_F$-compatible connection on $\E$ . We construct the real almost-commutative manifold $I^\infty_\bP\times_\nabla M$ as in \cref{defn:ACM}. Since $I^\infty_\bP$ is now a massive principal module (instead of a more general internal space), we will refer to $I^\infty_\bP\times_\nabla M$ as a \emph{principal almost-commutative manifold}. 

If $I^\infty_\bP$ is even with grading $\gamma_I$, we obtain a real even almost-commutative manifold $I^\infty_\bP\times_\nabla M$. Since the connection $\nabla$ is $G_F$-compatible, it automatically commutes with $J_I$ and $\gamma_I$ (see \cref{conn-comm}). Moreover, the same condition implies that the induced connection $[\nabla, \cdot]$ on $\End \bE$ restricts to $\bB$. It then follows from \cref{prop:spec_trip} that $I^\infty_\bP\times_\nabla M$ is a real even spectral triple of KO-dimension $4+k$ (mod $8$). 
\end{defn}

We continue in the remainder of this section, as in the usual approach for globally trivial almost-commutative manifolds (see \refscite{Connes96,CCM07} or the review \refcite{vdDvS12}), by generating the gauge fields and Higgs fields via inner fluctuations, and subsequently calculating the spectral action. 

\subsection{Inner fluctuations}
\label{sec:inn_fluc}
Let $(\B,\mH,D)$ be a spectral triple. Consider the \emph{generalised one-forms} given by
$$
\Omega_D^1(\B) := \Bigl\{ \sum_j a_j[D,b_j] \bigm\vert a_j,b_j\in\B \Bigr\} .
$$
For the canonical triple $(\A,L^2(\bS),\sD)$ of a spin manifold $M$, the generalised one-forms $\Omega_\sD^1(\A)$ are simply given by the Clifford multiplication $c$ of the usual one-forms $\Omega^1(M)$. To be precise, for smooth functions $f_1,f_2\in\A$, we obtain $f_1[\sD,f_2] = -if_1c(df_2)$. 

\begin{defn}
Let $(\B,\mH,D,J)$ be a real spectral triple. An \emph{inner fluctuation} of the operator $D$ is a self-adjoint element $A=A^*\in\Omega_D^1(\B)$. Such an inner fluctuation yields the \emph{fluctuated operator}
$$
D_A := D + A + \varepsilon'JAJ^* ,
$$
where the sign $\varepsilon'=\pm1$ is determined by the KO-dimension of the spectral triple (see \cref{defn:spectral_triple}). 
\end{defn}

For the remainder of this paper, we again assume that the dimension of $M$ is equal to $4$. We would like to show that, for a principal almost-commutative manifold, these inner fluctuations yield gauge fields and scalar fields (the latter are interpreted as Higgs fields in the noncommutative Standard Model). The inner fluctuations of the twisted Dirac operator $\sD_\bE := \1\otimes_\nabla\sD$ are (finite sums of) elements of the form
$$
a[\sD_\bE,b] = -i (\1\otimes c) \circ (a[\nabla,b]\otimes\1) ,
$$
for $a,b\in\B$, where $c$ denotes Clifford multiplication. The fact that $\nabla$ is a $G_F$-compatible connection ensures that $a[\nabla,b]\in\B\otimes_\A\Omega^1(M) \simeq \Omega^1(M,\bB)$. 
Requiring that $a[\sD_\bE,b]$ is self-adjoint then implies that $a[\nabla,b] \in \Omega^1(M,\lu(\bB))$, where $\lu(\bB)$ contains the anti-hermitian elements of $\bB$. An arbitrary inner fluctuation of $\sD_\bE$ is thus given by 
$$
\alpha := \sum_{j} a_j[\nabla,b_j] \in \Omega^1(M,\lu(\bB)) .
$$
We can then write $Ja[\sD_\bE,b]J^* = - i (\1\otimes c) \circ (J_I \alpha J_I^*\otimes\1)$, and consequently we have
$$
a[\sD_\bE,b] +  Ja[\sD_\bE,b]J^* = - i (\1\otimes c) \circ ( (\alpha + J_I \alpha J_I^*)\otimes\1).
$$

The inner fluctuations of the operator $D_I\otimes\gamma_5$ are of the form $\phi\otimes\gamma_5$, where
$$
\phi = \phi^* := \sum_j a_j[D_I,b_j] \in \Gamma^\infty(\End(\bE)).
$$

\begin{prop}
\label{prop:fluc}
The fluctuated Dirac operator $D_A := D + A + JAJ^*$ for a real even almost-commutative manifold is of the form
$$
D_A = 1\otimes_{\nabla'}\sD + \Phi\otimes\gamma_5 ,
$$
where $\nabla' := \nabla + \beta$ for some $\beta\in \Omega^1(M,\ad\bP)$, and $\Phi = \Phi^* := D_I + \phi + J_I\phi J_I^*\in\Gamma^\infty(\End(\bE))$ for some $\phi = \phi^* := \sum_j a_j[D_I,b_j]$. 
\end{prop}
\begin{proof}
The expression $\beta = \alpha+J_I\alpha J_I^*$ is an $\ad \bP$-valued $1$-form on $M$ (see \cref{lem:adP}). Noting that $\varepsilon'=1$ by assumption, the statement follows straightforwardly.
\end{proof}

The construction of $I^\infty_\bP\times_\nabla M$ explicitly uses the choice of a connection $\nabla$. However, we now show that this choice is irrelevant once we take the inner fluctuations into account. We need the following lemma.

\begin{lem}
\label{lem:omega1}
Let $\bB\to M$ be a unital $*$-algebra bundle, and let $\til\nabla$ be a connection on $\B = \Gamma^\infty(\bB)$ such that $\widetilde{\nabla} (1) = 0$, where $1$ denotes the identity section. 
Write $\A = C^\infty(M)$. Then 
\begin{align}
\label{eq:innerfluc}
\Big\{ \sum_j a_j \til\nabla(b_j) \bigm\vert a_j,b_j \in \B \Big\} = \B \otimes_\A \Omega^1(M) \simeq \Omega^1(M,\bB) . 
\end{align}
Consequently, $\Omega^1(M,\lu(\bB))$ is given by the anti-hermitian elements in $\big\{ \sum_j a_j \til\nabla(b_j) \mid a_j,b_j \in \B \big\}$. 
\end{lem}
\begin{proof}
Since $\til\nabla(b) \in \B\otimes_\A\Omega^1(M)$, the left hand side of \cref{eq:innerfluc} is clearly contained in the right hand side of \cref{eq:innerfluc}. For the converse inclusion, first suppose that both $a_j$ and $b_j$ are in $\A \subset Z(\B)$. In that case, 
\begin{align*}
\Big\{ \sum_j f_j \til\nabla(g_j \Id_\B) \bigm\vert f_j,g_j \in \A \Big\} \simeq \Big\{ \sum_j f_j dg_j \bigm\vert f_j,g_j \in \A \Big\} = \Omega^1(M). 
\end{align*}
It follows from this that
\begin{align*}
 \Big\{ \sum_j a_j \til\nabla(g_j 1) \bigm\vert a_j \in \B,g_j \in \A \Big\} = \B \otimes_\A \Omega^1(M). 
\end{align*}
Of course, the left-hand side of the previous equation is contained in $\big\{ \sum_j a_j \til\nabla(b_j) \mid a_j,b_j \in \B \big\}$, which proves the other inclusion.
\end{proof}

\begin{prop}
\label{prop:innerfluctuations}
Let $\bP\times_{G_F}F = (\B,\E,J_I)$ be a principal module over $M$ (for simplicity we consider here the massless case $D_I=0$) with two ($G_F$-compatible) connections $\nabla$ and $\nabla'$. 
Then $\1\otimes_{\nabla'}\sD$ is obtained as an inner fluctuation of $\1\otimes_\nabla\sD$.
\end{prop}
\begin{proof}
The difference between the two connections $\beta := \nabla' - \nabla$ is an element in $\Omega^1(M,\ad\bP)$. By \cref{lem:adP} there exists a (unique) element $\alpha \in \Omega^1(M,\lu) \subset \Omega^1(M,\lu(\bB))$ such that $\beta = \alpha + J_I\alpha J_I^*$.
The connection $\til\nabla = [\nabla, \cdot]$ on $\End(\E)$ restricts to a connection on $\B$, and satisfies $\widetilde{\nabla}(1) = 0$. 
\cref{lem:omega1} now implies that $\beta$ is obtained as an inner fluctuation.
\end{proof}

\begin{remark}
 We have seen that considering inner fluctuations of the Dirac operator essentially replaces the $G_F$-connection $\nabla$ (chosen in the construction of the almost-commutative manifold $I^\infty_\bP\times_\nabla M$) by a different (arbitrary) $G_F$-connection $\nabla'$. Therefore, after taking into account the inner fluctuations, our construction of principal almost-commutative manifolds is essentially independent of the initial choice of the connection $\nabla$. 
 
However, we also note that the endomorphisms $\Phi$ obtained through inner fluctuations in general remain dependent on the initial choice of $D_I$. 
\end{remark}

\subsection{The spectral action}
\label{ssct:spectralaction}

As mentioned immediately below \cref{defn:gauge_theory}, the dynamics of a gauge theory can be obtained from a gauge-invariant action functional. In the case of almost-commutative manifolds, such an action functional can be formulated in terms of the spectral triple. 

Let us first recall the definitions of the bosonic and fermionic action functionals for an arbitrary spectral triple $T = (\A,\mH,D)$. The bosonic part of the action functional is given by the \emph{spectral action} \cite{CC97}, defined as
$$
S_b(T) := \Tr\left( f \left( \frac{D_A}{\Lambda} \right) \right) .
$$
Here $\Tr$ denotes the operator trace on $B(\mH)$, $D_A$ is the fluctuated Dirac operator, $f\colon \mathbb{R} \rightarrow \mathbb{R}$ is some positive even function, and $\Lambda\in\R$ is a (large) cut-off parameter. The function $f$ is assumed to decay sufficiently rapidly at infinity so that the trace of $f(D_A/\Lambda)$ exists. In particular, $f$ could be considered as a smooth approximation to a cut-off function (and as such it counts the number of eigenvalues of $D_A$ whose absolute values are smaller than $\Lambda$), but this viewpoint is not necessary for the following.

If the spectral triple is even (with grading $\gamma$) and has a real structure $J$ of KO-dimension $2$, the \emph{fermionic action} \cite{Connes06} is defined as
$$
S_f(T) := \frac12 \langle J\til\xi,D_A\til\xi\rangle ,
$$
where $\til\xi$ is the Grassmann variable corresponding to a vector $\xi\in\mH^+$ (i.e.\ $\gamma\xi=\xi$). 

We quote the following well-known result:
\begin{prop}[see e.g.\ {\refcite[\S2.6.1]{vdDvS12}}]
\label{prop:gauge_invariant}
For a real spectral triple $T = (\A,\mH,D,J,\gamma)$ of KO-dimension $2$, the action functionals $S_b(T)$ and $S_f(T)$ are invariant under the action of the gauge group $\G(T)$. 
\end{prop}

We now provide explicit formulas for the spectral action of principal almost-commutative manifolds (formulas for the fermionic action will only be given for the example of electrodynamics in \cref{sec:ex_ED}). The spectral action was calculated in \refscite{CCM07,Connes06} for the product triple $M\times F$, where $F$ was chosen in order to describe the full Standard Model of elementary particle physics. In the remainder of this section we largely follow the notation of \refcite{vdDvS12}, where also detailed derivations of the formulas provided here can be found. 

For the canonical triple $(C^\infty(M),L^2(\bS),\sD)$ of a smooth compact $4$-dimensional Riemannian spin manifold $M$, the spectral action yields the asymptotic formula 
\begin{equation*}
S_b(M) \sim_{\Lambda\to\infty} \int_M \La_M(g_{\mu\nu}) \sqrt{|g|} d^4x + \mO(\Lambda^{-1}) ,
\end{equation*}
where $g$ is the Riemannian metric on $M$. 
The Lagrangian $\La_M$ is given by 
\begin{align}
\label{eq:canon_Lagr}
\La_M(g_{\mu\nu}) := \frac{f_4\Lambda^4}{2\pi^2} - \frac{f_2\Lambda^2}{24\pi^2} s + \frac{f(0)}{16\pi^2} \Big(\frac1{30} \Delta s - \frac1{20} C_{\mu\nu\rho\sigma} C^{\mu\nu\rho\sigma} + \frac{11}{360}R^*R^* \Big) .
\end{align}
Here $s$ denotes the scalar curvature of $M$, $\Delta$ is the scalar Laplacian, $C$ is the Weyl curvature, and $R^*R^*$ is a topological term, which integrates to (a multiple of) the Euler class. The coefficients $f_k$ (for $k>0$) are the moments of $f$, defined as
$$
f_k := \int_0^\infty f(t) t^{k-1} dt .
$$

We now provide the spectral action for a principal almost-commutative manifold. As all calculations are local, the result is exactly the same as for the spectral action of a product triple $M\times F$, and we refer to \refcite{vdDvS12} for the detailed calculations. 

In \cref{prop:fluc}  we saw that the fluctuated Dirac operator is determined by a connection $\nabla'= \nabla + \beta$ and an endomorphism $\Phi$ on $\bE$. 
From here on we shall work on a local trivialisation $(U,h\Sub{U})$, where we can write $\nabla|_U = h\Sub{U}^{-1} \circ (d + \omega\Sub{U}) \circ h\Sub{U}$, and define the local $\g_F$-valued $1$-form $B := \omega\Sub{U} + h\Sub{U}\circ\beta|_U\circ h\Sub{U}^{-1} \in \Omega^1(U,\g_F)$ (for ease of notation we do not make the dependence of $B$ on the local chart $U$ explicit). Thus $B$ is the local connection form for $\nabla'$. Using a local coordinate basis $\partial_\mu$, we define $B_\mu := B(\partial_\mu) \in C^\infty(U,\g_F)$. We omit the local trivialisation $h\Sub{U}$ from our notation, so we write e.g.\ $\nabla'_\mu = \partial_\mu + B_\mu$. Furthermore, we introduce the notation 
\begin{align*}
D_\mu\Phi &:= [\nabla'_\mu,\Phi] = \partial_\mu\Phi + [B_\mu,\Phi] , &
F_{\mu\nu} &:= \partial_\mu B_\nu - \partial_\nu B_\mu + [B_\mu,B_\nu] .
\end{align*}

\begin{prop}
\label{prop:spectralaction}
The spectral action for a principal almost-commutative manifold $I^\infty_\bP\times_\nabla M$ is asymptotically given by the local formula
$$
S_b(I^\infty_\bP\times_\nabla M) \sim_{\Lambda\to\infty} \int_M \La(g_{\mu\nu},B_\mu,\Phi) \sqrt{|g|} d^4x + \mO(\Lambda^{-1}) ,
$$
for
\begin{align*}
\La(g_{\mu\nu}, B_\mu, \Phi) := N \La_M(g_{\mu\nu}) + \La_B(g_{\mu\nu},B_\mu) + \La_H(g_{\mu\nu}, B_\mu, \Phi) .
\end{align*}
Here $\La_M(g_{\mu\nu})$ is given in \cref{eq:canon_Lagr}, and $N$ is the rank of $\bE$. $\La_B$ gives the kinetic term of the gauge field and equals
\begin{align*}
\La_B(g_{\mu\nu},B_\mu) := \frac{f(0)}{24\pi^2} \tr(F_{\mu\nu}F^{\mu\nu}) ,
\end{align*}
where $\tr$ denotes the fibre-wise trace for endomorphisms on the bundle $\bE\otimes \bS$.
$\La_H$ gives the Higgs Lagrangian 
given by
\begin{align*}
\La_H(g_{\mu\nu}, B_\mu, \Phi) &:= -\frac{2f_2\Lambda^2}{4\pi^2} \tr(\Phi^2) + \frac{f(0)}{8\pi^2} \tr(\Phi^4) + \frac{f(0)}{24\pi^2} \Delta\big(\tr(\Phi^2)\big) \\
&\qquad\qquad\qquad\qquad\qquad\qquad+ \frac{f(0)}{48\pi^2} s\tr(\Phi^2) + \frac{f(0)}{8\pi^2} \tr\big((D_\mu \Phi)(D^\mu \Phi)\big) ,
\end{align*}
where the first two terms form the Higgs potential, the third is a boundary term, the fourth couples the Higgs field to the scalar curvature, and finally we have the kinetic term including interactions with the gauge field. 
\end{prop}

\begin{remark}
Although the above explicit formulas for the spectral action are exactly the same as for a product triple $M\times F$, there can nonetheless be a significant difference, because the constant matrix $D_F$ is replaced by a global endomorphism $D_I$. For a product triple $M\times F$, the inner fluctuations of $\gamma_5\otimes D_F$ also lead to global endomorphisms of the form $\gamma_5\otimes\Phi$, where $\Phi\in\Gamma^\infty(\End(\bE))$ (though this $\Phi$ would be more restricted than in our construction). 
However, there may be components of $D_F$ that are not affected by inner fluctuations, and hence remain constant (this occurs for instance for the Majorana masses of right-handed neutrinos in the case of the noncommutative Standard Model \cite{CCM07}). In the case of a principal almost-commutative manifold, these components could be non-constant from the start. Hence, compared to the case of product triples, derivatives of the field $\Phi$ might contain additional terms. This difference is not yet visible in the general formulas above, but it may have consequences once we look at concrete examples (see \cref{rem:var_mass}). 
\end{remark}

\subsection{Gauge theory}

The results of this section can be summarised as follows, which is the main result of our paper:

\begin{thm}
\label{thm:main}
Let $M$ be a smooth compact $4$-dimensional Riemannian spin manifold. Consider a massive even principal module $I^\infty_\bP = (\B, \E, D_I, \gamma_I, J_I)$ of $KO$-dimension $k$ over $M$. Let $\nabla$ be a $G_F$-compatible connection on $\E$. If $M$ is simply connected, then the principal almost-commutative manifold $I^\infty_\bP\times_\nabla M$ of $KO$-dimension $4 + k$ (mod $8$) describes a classical gauge theory over $M$ with gauge group $\G(I^\infty_\bP\times_\nabla M)$. 
\end{thm}
\begin{proof}
The principal module $I^\infty_\bP$ is constructed from a principal $G_F$-bundle $\bP$ over $M$, such that $\B$ and $\E$ are given by smooth sections of bundles associated to $\bP$. 
By assumption $M$ is simply connected, so it follows from \cref{coro:gaugegroupsisomorphic} that we have the isomorphism $\G(I^\infty_\bP\times_\nabla M) \simeq \G(\bP)$. We have seen in \cref{sec:inn_fluc} that the inner fluctuations transform a $G_F$-compatible connection on $\E$ to another $G_F$-compatible connection, which hence corresponds to a connection on $\bP$ (and by \cref{prop:innerfluctuations} any connection on $\bP$ can be obtained in this way). Finally, the spectral action and the fermionic action provide a gauge-invariant action functional (see \cref{prop:gauge_invariant}). 
Thus the principal almost-commutative manifold $I^\infty_\bP\times_\nabla M$ provides all the necessary ingredients for a classical gauge theory over $M$, as described in \cref{defn:gauge_theory}. 
\end{proof}

\section{Examples}
\label{sec:examples}

In this section we adapt two simple examples of (globally trivial) gauge theories in the context of noncommutative geometry to the globally non-trivial case. In each example, we assume (as before) that the underlying manifold $M$ is a smooth compact $4$-dimensional Riemannian spin manifold. 

In \cref{sec:ex_YM} we describe the Yang-Mills case that was studied in \refcite{BvS11}, and provided the motivation for this work. In particular, we show that the Yang-Mills case provides examples of principal modules that cannot be described by gauge modules. 
In \cref{sec:ex_ED} we discuss the abelian gauge theory of electrodynamics, based on the (globally trivial) description in \refcite{vdDvS13}. We will describe the resulting (globally non-trivial) gauge theory, and provide explicit formulas for both the spectral action and the fermionic action. 

\subsection{Yang-Mills}
\label{sec:ex_YM}

Globally trivial Yang-Mills theory was already studied  in the setting of spectral triples by Chamseddine and Connes \cite{CC97}. It is described by the (real, even) finite spectral triple 
\begin{align*}
F\Sub{\YM} :=  (M_N(\mathbb{C}), M_N(\mathbb{C}), D_F=0 , J_F = (\cdot)^*, \gamma_F = \text{id}),
\end{align*}
where the algebra $M_N(\mathbb{C})$ acts on the Hilbert space $M_N(\mathbb{C})$ by left-multiplication. The $KO$-dimension of this spectral triple is $0$ and the structure group $G_F$ is equal to $PSU(N)$. 

This has been generalised to the globally non-trivial case in \refcite{BvS11}. Let $\bB\to M$ be an arbitrary $*$-algebra bundle with fibre $M_N(\mathbb{C})$, and let $\B = \Gamma^\infty(\bB)$ be its unital, involutive $C^\infty(M)$-module algebra of sections. We consider the real even internal space
$$
I^\infty\Sub{\YM} := (\B,\B,D_I=0,J_I=(\cdot)^*,\gamma_I=\text{id}) .
$$

For a general principal module $\bP \times_{G_F} F$ we do not know how to reconstruct the principal bundle $\bP$ from the module. However, in the Yang-Mills case we do.

\begin{lem}
 There exists a principal $PSU(N)$-bundle $\bP\to M$ (unique up to isomorphism) such that $I^\infty\Sub{\YM} \simeq \bP\times_{PSU(N)}F\Sub{\YM}$.
\end{lem}
\begin{proof}
The transition functions of the $*$-algebra bundle $\bB$ take values in $\text{Aut}(M_N(\mathbb{C})) \simeq PSU(N)$ (where $PSU(N)$ acts on $M_N(\C)$ by conjugation). Hence by \cref{thm:reconstruction} we can reconstruct a principal $PSU(N)$-bundle $\bP$ such that $\bB \simeq \bP\times_{PSU(N)}M_N(\C)$. Since $PSU(N)$ is the full automorphism group of the fibre, the bundle $\bP$ is uniquely defined. 
\end{proof}

\begin{remark}
 Note that $I^\infty\Sub{\YM}$ will in general \emph{not} be a gauge module. If this were the case, the structure group $PSU(N)$ of $\bB$ could be lifted to $U(N)$ by \cref{prop:gaugemodulestructuregroup}. 
This is only possible if the Dixmier-Douady class $\delta(\bB) \in \check H^3(M,\mathbb{Z})$ is identically zero (see e.g.\ \refcite[Ch.5]{RW98} or \refcite{Sch09} for more details on Dixmier-Douady classes), which is equivalent to saying that $\bB$ is an endomorphism bundle 
(note that this is consistent with the condition $\bB_i = \End(\bE_i)$ in \cref{defn:gaugemodule}). 
Since not every $*$-algebra bundle with fibre $M_N(\mathbb{C})$ has zero Dixmier-Douady class (see e.g.\ \refcite{Sch09}), this example shows that there exist principal modules that are not gauge modules. However, in our description of gauge theories in \cref{sec:gauge_theory} we have restricted our attention to simply connected, $4$-dimensional manifolds, and it turns out that in this case the Dixmier-Douady class always vanishes (as we will prove below). It is unclear if there exist other examples of principal modules that are not gauge modules. 
\end{remark}

\begin{prop}
Let $\bB$ be a $*$-algebra bundle with fibre $M_N(\C)$ over a simply connected, $4$-dimensional, oriented, compact manifold $M$. Then the Dixmier-Douady class of $\bB$ is identically zero.
\end{prop}
\begin{proof}
Since $M$ is simply connected, its fundamental group is trivial, and hence (see e.g.\ \refcite[Theorem 2.A.1]{Hatcher02}) the first singular homology group $H_1(M,\Z)$ is trivial. 
By Poincar\'e duality (see e.g.\ \refcite[Proposition 3.25 \& Theorem 3.30]{Hatcher02}) it then follows that the third cohomology group $H^3(M,\Z)$ is also trivial. 
The Dixmier-Douady class by definition takes values in the third \v Cech cohomology group $\check H^3(M,\Z)$. Since for compact manifolds these cohomology groups are equal, it follows that $\check H^3(M,\Z)$ is trivial and hence that the Dixmier-Douady class of $\bB$ must vanish.
\end{proof}

A connection $\nabla\colon\B \to \B\otimes_\A\Omega^1(M)$ is $PSU(N)$-compatible (\textit{cf.} \cref{sct:pacm}) if and only if it satisfies the algebraic identities (see \refcite[\S3.2]{BvS11})
\begin{align*}
\nabla(ab) &= \nabla(a)b + a\nabla(b) , & (\nabla a)^* &= \nabla(a^*) , & \forall a,b\in\B .
\end{align*}
Such a connection thus corresponds to a connection form $\omega$ on $\bP$. 
If we pick any such connection, we can then consider the (principal) almost-commutative manifold
\begin{align*}
I^\infty\Sub{\YM} \times_\nabla M := 
\big( \Gamma^\infty(\bB) , L^2(\bB\otimes\bS) , \sD_\bB , J_I\otimes J_M , \gamma_I\otimes\gamma_5 \big) .
\end{align*}
If $M$ is simply connected, the group $\G(I^\infty\Sub{\YM}\times_\nabla M)$ is isomorphic to $\G(\bP)$, and $I^\infty\Sub{\YM} \times_\nabla M$ describes a $PSU(N)$ gauge theory $(\bP,\omega)$ over $M$. We denote the local connection form of $\nabla$ by $B_\mu$, and its curvature tensor by $F_{\mu\nu}$. From \cref{prop:spectralaction} we find that the spectral action yields the Lagrangian
$$
\La(g_{\mu\nu},B_\mu) = N^2 \La_M(g_{\mu\nu}) + \La\Sub\YM(g_{\mu\nu},B_\mu) ,
$$
where the Yang-Mills Lagrangian is given (up to a normalisation constant) by the usual expression:
$$
\La\Sub\YM(g_{\mu\nu},B_\mu) := \frac{f(0)}{24\pi^2} \tr(F_{\mu\nu}F^{\mu\nu}) .
$$

\subsection{Electrodynamics}
\label{sec:ex_ED}
The example of (globally trivial) Electrodynamics in the context of noncommutative geometry appeared in \refcite{vdDvS13}. Here we describe its generalisation to the globally non-trivial case. The finite spectral triple for electrodynamics is given by \cite{vdDvS13}
$$
F\Sub{\ED} := (\C^2, \C^4, D_F, \gamma_F, J_F) .
$$
We shall generalise this finite triple to a massive even gauge module $I^\infty\Sub{\ED}$ over $M$. First, we set the algebra $\B$ to be of the form
$$
\B := \A\oplus\A = C^\infty(M) \oplus C^\infty(M) .
$$
Let $\bL$ be a complex line bundle over $M$, with a given hermitian structure, so that its structure group is $U(1)$. 
We shall take two identical copies of this line bundle, which we denote by $\bE_L$ and $\bE_R$, with smooth sections $\E_L = \Gamma^\infty(\bE_L)$ and $\E_R = \Gamma^\infty(\bE_R)$. Then the Hilbert $\B-\A$-bimodule $\E$ is defined as 
$$
\E := (\E_L\oplus \E_R) \oplus (\overline{\E_L}\oplus \overline{\E_R}) ,
$$
where the first component of $\B$ acts on $\E_L\oplus \E_R$, and the second component acts on its conjugate. On this decomposition, the grading is defined as $\gamma_I := 1\oplus(-1)\oplus(-1)\oplus1$. The real structure $J_I$ is the anti-linear map $\E_{L,R}\mapsto\bar{\E_{L,R}}$ and $\bar{\E_{L,R}}\mapsto \E_{L,R}$ of KO-dimension $6$ (see \cref{defn:spectral_triple}). We then have the subalgebra $\B_J \simeq \A \subset \B$, where the injection is given by $a\mapsto a\oplus a$. Imposing all conditions in \cref{defn:mass_matrix}, the `mass matrix' $D_I$ is restricted to be of the form 
$$
D_I := \matfour{0&d&0&0}{\bar d&0&0&0}{0&0&0&\bar d}{0&0&d&0} ,
$$
where $d\in C^\infty(M)$ (see \refcite[\S4.1.1]{vdDvS13}). 

\begin{remark}
\label{rem:var_mass}
In order to interpret $d$ as a mass parameter, it would have to be given by a single real-valued parameter. For this reason we restrict ourselves to the case $d = -im$ (see \refcite[Remark 4.4]{vdDvS13}). We stress here that in general (as mentioned in \cref{remark:varying_mass}) the mass $m$ is not a fixed parameter, but a function on $M$ (although it can be chosen to be constant). In other words, our framework allows the mass of a particle to vary from point to point in $M$, so essentially the Yukuwa mass parameter is replaced by a Yukawa field. This could of course have significant physical implications, which we intend to study in future work. 
\end{remark}

The module $I^\infty\Sub{\ED} = (\B,\E,D_I,\gamma_I,J_I)$ defined in this way is in fact a massive even gauge module. To be precise, if we write $\E_1 := \Gamma^\infty(\bL) = \E_L = \E_R$ and $\E_2 := \A$, then we have $\B_1 = \End_\A(\E_1) = \Gamma^\infty(\bL\otimes\bL^*) \simeq \A$ and also $\B_2 \simeq \A$. Furthermore, the module $\E$ can be written as 
\begin{align*}
 \E &\simeq \bigoplus_{(i,j)\in K} \E_i\otimes_\A\bar{\E_j} , & K &:= \big\{ (1,2),(1,2),(2,1),(2,1) \big\} .
\end{align*}

The hermitian structure on $\bL$ determines a class of transition functions of $\bL$ taking values in $U(1)$, so using \cref{thm:reconstruction} we can uniquely reconstruct a principal $U(1)$-bundle $\bP$, and we have $I^\infty\Sub{\ED} \simeq \bP\times_{U(1)}F\Sub{\ED}$ as \emph{massless} modules (i.e.\ ignoring the mass matrices $D_F$ and $D_I$).

\begin{prop}
The gauge group is given by
$$
\G(I^\infty\Sub{\ED}) \simeq \mU(\B) / \mU(\B_J) \simeq \Gamma^\infty(\Ad\bP) \simeq C^\infty(M,U(1)) .
$$
\end{prop}
\begin{proof}
Note that the group bundle $\Ad\bP \simeq M\times U(1)$ is globally trivial, because the structure group $U(1)$ is abelian. 

As in \cref{sssct:gaugegroup}, the main thing to prove is the surjectivity of the map $\phi_*\colon \mU(\B)\to \Gamma^\infty(\Ad\bP)$, which is given by $\phi_*(u) = uJuJ^*$. But for $u=(u_1,u_2)\in\mU(\B)$, this map is given by
$$
(u_1,u_2) \mapsto 
\mattwo{u_1u_2^*}{0}{0}{u_2u_1^*} ,
$$
so $\phi_*(u_1,u_2)$ can be identified with $u_1u_2^*$. Hence each $v\in \Gamma^\infty(\Ad\bP) \simeq C^\infty(M,U(1))$ is the image of $(v,1)\in\mU(\B)$. 
\end{proof}

\begin{remark}
Note that in this particular example it is not necessary to require that $M$ is simply connected, as we did in the general case (see \cref{coro:gaugegroupsisomorphic}). 
\end{remark}

 An element $\lambda\in \G(I^\infty\Sub{\ED})$ acts on $\E_L\oplus\E_R$ as multiplication by $\lambda$, and acts on $\overline{\E_L}\oplus \overline{\E_R}$ as multiplication by $\bar\lambda$. 

Pick a connection $\nabla^\bL$ on $\bL$, and let the connection $\nabla$ on $\E$ be given by
$$
\nabla := \nabla^\bL \oplus \nabla^\bL \oplus \bar{\nabla^\bL} \oplus \bar{\nabla^\bL} .
$$
On a local trivialisation (say on a neighbourhood $U$), the connection $\nabla^\bL$ is determined by a local connection form $\omega^\bL\Sub{U} \in \Omega^1(U,i\R)$, where $i\R$ is the Lie algebra of $U(1)$. 
For the connection $\nabla$ on $\E$ this yields the connection form
$$
\omega\Sub{U} = \omega^\bL\Sub{U} \oplus \omega^\bL\Sub{U} \oplus \bar{\omega^\bL\Sub{U}} \oplus \bar{\omega^\bL\Sub{U}} = \omega^\bL\Sub{U} \left( 1\oplus1\oplus(-1)\oplus(-1) \right) ,
$$
where the last equality follows because the action of $\bar{\omega^\bL\Sub{U}}$ is given by (right) multiplication with ${\omega^\bL\Sub{U}}^* = -\omega^\bL\Sub{U}$. 

Now consider the almost-commutative manifold $I^\infty\Sub{\ED}\times_\nabla M$ of KO-dimension $2$, which (by \cref{thm:main}) describes a $U(1)$-gauge theory over $M$. 
Taking inner fluctuations simply amounts to choosing a different connection $\nabla^\bL$ (see \cref{prop:innerfluctuations}), while there is no Higgs field (because $D_I$ commutes with $\B$). Hence we ignore these inner fluctuations, and simply consider the local gauge field $A_\mu := \omega^\bL\Sub{U}(\partial_\mu)$, on some coordinate basis $\partial_\mu$. 
Its curvature is defined as $\mF_{\mu\nu} := \partial_\mu A_\nu - \partial_\nu A_\mu$. From \cref{prop:spectralaction} (see also \refcite[Proposition 4.2]{vdDvS13}) we find that the spectral action for $I^\infty\Sub{\ED}\times_\nabla M$ is asymptotically given by the local formula
$$
S_b(I^\infty\Sub{\ED}\times_\nabla M) \sim_{\Lambda\to\infty} \int_M \La(g_{\mu\nu},A_\mu,m) \sqrt{|g|} d^4x + \mO(\Lambda^{-1}) ,
$$
for
\begin{align*}
\La(g_{\mu\nu}, A_\mu, m) := 4 \La_M(g_{\mu\nu}) + \La_A(g_{\mu\nu},A_\mu) + \La_H(g_{\mu\nu},m) .
\end{align*}
Here $\La_M(g_{\mu\nu})$ is the Lagrangian \cref{eq:canon_Lagr}, and $\La_H(g_{\mu\nu},m)$ yields additional terms depending on the mass $m$ and the scalar curvature $s$:
$$
\La_H(g_{\mu\nu},m) := -\frac{2 f_2 \Lambda^2 m^2}{\pi^2} + \frac{ f(0) m^4}{2\pi^2} + \frac{f(0) m^2 s }{12\pi^2}.
$$
The Lagrangian for the gauge field is given by
\begin{align*}
\La_A(g_{\mu\nu},A_\mu) := \frac{f(0)}{6\pi^2} \mF_{\mu\nu}\mF^{\mu\nu} .
\end{align*}

The interaction of the $U(1)$ gauge field with the fermions is described by the fermionic action, and is given by (see \refcite[Proposition 4.3 and Theorem 4.5]{vdDvS13})
$$
S_f(I^\infty\Sub{\ED}\times_\nabla M) = \int_M \La_f(g_{\mu\nu},A_\mu,m) \sqrt{|g|} d^4x ,
$$
for the Lagrangian
$$
\La_f(g_{\mu\nu},A_\mu,m) := -i \left( J_M\til\chi , \big(\gamma^\mu(\nabla^S_\mu-A_\mu) - m\big)\til\psi \right) ,
$$
where $\chi$ and $\psi$ are two Dirac spinors in $L^2(\bS)$. 
We summarise this as follows:
\begin{prop}
The total Lagrangian for $I^\infty\Sub{\ED}\times_\nabla M$ is given by a gravitational part 
$$
\La_{\rm grav}(g_{\mu\nu},m) := 4 \La_M(g_{\mu\nu}) + \La_H(g_{\mu\nu},m) ,
$$
and a part for electrodynamics
$$
\La\Sub{\ED}(g_{\mu\nu}, A_\mu, m) := -i \left( J_M\til\chi , \big(\gamma^\mu(\nabla^S_\mu-A_\mu) - m\big)\til\psi \right) + \frac{f(0)}{6\pi^2} \mF_{\mu\nu}\mF^{\mu\nu} .
$$
\end{prop}

\section{Outlook}
\label{sec:outlook}

One of the main ideas in the development of noncommutative geometry has been the translation of geometric data into (operator-)algebraic data. In this light, it is somewhat unsatisfactory that our definition of principal modules relies entirely on the geometric notion of a principal bundle. Our discussion of gauge modules is an attempt to provide a purely algebraic approach, but as we have shown, these gauge modules only yield a proper subclass of principal modules. 
It is still an open question how arbitrary principal modules should be described algebraically, that is, what algebraic structure on a triplet $(\B,\E,J)$ would completely characterise the properties of a principal module. 
The decompositions $\E = \oplus_{i,j \in I} \E_{ij}$ and $\B = \oplus \B_i$ (as described in \cref{sssct:gaugegroup}) are not yet enough to ensure that $(\B,\E,J)$ is a principal module. On the other hand, the condition that $\E_{ij} = \E_i \otimes_\A \bar{\E}_j$ (modulo multiplicities) along with $\B_i = \End(\E_i)$, as for gauge modules, is in fact too strong. 

As mentioned in \cref{rem:bundle-atlas}, the principal bundle $\bP$ can only be reconstructed from the associated vector bundle $\bE = \bP\times_{G_F}\mH_F$ if we also know the corresponding equivalence class of $G_F$-atlases. It is not clear if there exists a geometric structure on $\bE$, for which this equivalence class corresponds precisely to those transition functions that preserve the geometric structure.  
If one has such a geometric structure on $\bE$, this might provide the possibility of finding an algebraic equivalent structure on the module $\E$. 
We intend to return to these questions in the future. 


In \cref{sec:examples} we described two basic examples, namely Yang-Mills theory and electrodynamics. It would of course be more interesting to also put the description of the noncommutative Standard Model \cite{CCM07} into our globally non-trivial framework. This should certainly be possible, though it would require some small modifications to accommodate real algebras (in this paper we have always assumed that our algebras are complex). In particular, for real algebras the resulting gauge group would not automatically be unimodular (see also \cref{remark:unimod}), and one would have to impose unimodularity by hand (as in \refcite[\S2.5]{CCM07}). 
More importantly, as we also mentioned in \cref{remark:varying_mass,rem:var_mass}, in our framework the mass parameters (i.e.\ the Yukawa couplings and the Majorana terms) of the theory are not restricted to be constant, but they are allowed to vary on spacetime. Such variation of the Majorana mass then naturally leads to a new scalar field $\sigma$, which was used in \refcite{CC12} to restore the consistency of the noncommutative Standard Model with the experimental value of the Higgs mass. In addition however, the variation of the Yukawa couplings will also have its effect on the physical theory. We hope to provide a more detailed study of these physical implications in a future work.

\section*{Acknowledgements}

The second author wishes to thank the hospitality of the Radboud University Nijmegen during two short visits in July 2012 and September 2013. The second author also thanks the hospitality of Leipzig University for a long-term visit from September 2013 to January 2014, when part of this work was done. Both authors are grateful to Adam Rennie, Walter van Suijlekom, Simon Brain, Eli Hawkins and Klaas Landsman for helpful suggestions and discussions. The first author acknowledges support from NWO via the GQT-cluster. The second author acknowledges support from both the Australian National University and the University of Wollongong.


\providecommand{\noopsort}[1]{}\providecommand{\vannoopsort}[1]{}
\providecommand{\bysame}{\leavevmode\hbox to3em{\hrulefill}\thinspace}
\providecommand{\MR}{\relax\ifhmode\unskip\space\fi MR }
\providecommand{\MRhref}[2]{%
  \href{http://www.ams.org/mathscinet-getitem?mr=#1}{#2}
}
\providecommand{\href}[2]{#2}

\end{document}